\documentclass[11pt]{article}
\usepackage[T1]{fontenc}
\usepackage{amsfonts,amsmath,amsthm,amssymb,dsfont,mathtools}
\usepackage[margin=1in]{geometry}
\usepackage{fullpage}
\usepackage{enumitem}
\usepackage[linesnumbered,boxed,ruled,vlined]{algorithm2e}
\usepackage[colorlinks,citecolor=blue,linkcolor=blue,urlcolor=red,pagebackref]{hyperref}
\usepackage[capitalise]{cleveref}
\usepackage{url}
\usepackage{thmtools, thm-restate}
\usepackage{xcolor}
\usepackage{graphicx}
\usepackage{mathdots}
\usepackage[normalem]{ulem}
\usepackage{xspace}
\usepackage{tabularx}
\usepackage{array}
\newcolumntype{Y}{>{\raggedright\arraybackslash}X}
\xspaceaddexceptions{]\}}
\usepackage{comment}
\usepackage{framed}
\usepackage{float,wrapfig}
\usepackage{subfigure}
\usepackage{tikz}
\usetikzlibrary{arrows.meta, positioning, calc, fit, backgrounds}

\usetikzlibrary{patterns}
\usepackage[usenames,dvipsnames]{pstricks}

\usepackage[mathlines]{lineno}
\usepackage{etoolbox}
\newcommand*\linenomathpatch[1]{%
  \cspreto{#1}{\linenomath}%
  \cspreto{#1*}{\linenomath}%
  \csappto{end#1}{\endlinenomath}%
  \csappto{end#1*}{\endlinenomath}%
}
\linenomathpatch{equation}
\linenomathpatch{gather}
\linenomathpatch{multline}
\linenomathpatch{align}
\linenomathpatch{alignat}
\linenomathpatch{flalign}

\theoremstyle{plain}
\newtheorem{property}{Property}
\newtheorem{theorem}{Theorem}[section]
\newtheorem{lemma}[theorem]{Lemma}
\newtheorem{proposition}[theorem]{Proposition}
\newtheorem{corollary}[theorem]{Corollary}
\newtheorem{claim}[theorem]{Claim}
\newtheorem{observation}[theorem]{Observation}

\theoremstyle{definition}
\newtheorem{definition}[theorem]{Definition}

\crefname{property}{Property}{Properties}

\setlist[enumerate]{nosep, topsep=1ex}
\setlist[itemize]{nosep, topsep=1ex}
\setlist[description]{nosep}
\allowdisplaybreaks
\def\ShowAuthNotes{1}
\ifnum\ShowAuthNotes=1
\newcommand{\authnote}[2]{\ \\ \textcolor{red}{\parbox{0.9\linewidth}{[{\footnotesize {\bf #1:} { {#2}}}]}}\newline}
\else
\newcommand{\authnote}[2]{}
\fi

\newcommand{\eps}{\epsilon}

\newcommand{\Ex}{\operatorname*{\mathbf{E}}}

\newcommand{\poly}{\operatorname{\mathrm{poly}}}
\newcommand{\polylog}{\operatorname{\mathrm{polylog}}}

\newcommand{\inse}{{\textsc{Insert}}}
\newcommand{\dele}{{\textsc{Delete}}}
\newcommand{\rebu}{{\textsc{Rebuild}}}
\newcommand{\crea}{{\textsc{CreateBundles}}}
\newcommand{\lowlevel}{{\mathrm{leftlevel}}}
\newcommand{\typ}{{\mathrm{type}}}
\newcommand{\sz}{\mu}

\newcommand{\Z}{\mathbb{Z}}

\newcommand{\caA}{\mathcal{A}}
\newcommand{\caB}{\mathcal{B}}

\newcommand{\caS}{\mathcal{S}}
\newcommand{\caT}{\mathcal{T}}

\title{Memory Reallocation with Polylogarithmic Overhead}
\author{Ce Jin\thanks{\url{cejin@berkeley.edu}. 
This work is supported by the Miller Research Fellowship at the Miller Institute for Basic Research in Science, UC Berkeley.}\\UC Berkeley}
\date{\vspace{-1cm}}

\begin{document}

	\maketitle

\begin{abstract}
The \emph{Memory Reallocation problem} asks to dynamically maintain an assignment of given objects of various sizes to non-overlapping contiguous chunks of memory, while supporting updates (insertions/deletions) in an online fashion. The total size of live objects at any time is guaranteed to be at most a $1-\epsilon$ fraction of the total memory. To handle an online update, the allocator may rearrange the objects in memory to make space, and the \emph{overhead} for this update is defined as the total size of moved objects divided by the size of the object being inserted/deleted.

Our main result is an allocator with worst-case expected overhead $\mathrm{polylog}(\epsilon^{-1})$. This exponentially improves the previous worst-case expected overhead $\widetilde O({\epsilon}^{-1/2})$ achieved by Farach-Colton, Kuszmaul, Sheffield, and Westover (2024), narrowing the gap towards the $\Omega(\log\epsilon^{-1})$ lower bound. Our improvement is based on an application of the sunflower lemma previously used by Erd\H{o}s and S\'{a}rk\"{o}zy (1992) in the context of subset sums. 

Our allocator achieves polylogarithmic overhead only in expectation, and sometimes performs expensive rebuilds. Our second technical result shows that this is necessary: it is impossible to achieve subpolynomial overhead with high probability. 
\end{abstract}

\section{Introduction}

Memory allocation is one of the oldest problems in computer science, both in theory and in practice.

Consider an array of $M$ memory slots indexed by $\{0,1,\dots,M-1\}$, and a dynamically changing collection of \textbf{\emph{objects}} $X$,\footnote{\emph{Objects} are also referred to as memory requests, items, or blocks in the literature.} where object $x\in X$ has \textbf{\emph{size}} $\mu(x) \in \Z^{+}$. 
An \textbf{\emph{allocation}} is an assignment of the objects to non-overlapping contiguous regions in memory, or formally, an allocation map $\phi\colon X \to \{0,1,\dots,M-1\}$ such that $\phi(x)+\mu(x)\le M$, and the intervals $[\phi(x),\phi(x)+\mu(x))$ are disjoint for all $x\in X$.
An \textbf{\emph{allocator}} needs to maintain the allocation while handling two types of \textbf{\emph{updates}} to $X$ in an online fashion (starting from the initial state $X=\varnothing$):
\begin{itemize}
   \item \textbf{\emph{Inserting}} a new object $x$ of size $\mu(x)$ into $X$ (memory allocation).
   \item \textbf{\emph{Deleting}} an object $x\in X$ from $X$ (memory deallocation / free).
\end{itemize}
We sometimes refer to objects currently in $X$ as \textbf{\emph{live objects}} for emphasis.

In the classic setting of memory allocation, objects cannot be moved after being allocated. 
Upon inserting an object $x$, the allocator decides its location $\phi(x)$, and $x$ continues to occupy the memory interval $[\phi(x),\phi(x)+\mu(x))$ until its deletion.
After a few updates, the free space in memory may become fragmented. When inserting a new object~$x$, even though the total empty space is larger than $\mu(x)$, there may be no contiguous region that fits it.
This causes a waste of memory.

Unfortunately, classic results show that any online allocator must inevitably waste \emph{most of} the memory in this no-move setting.
Formally, we say an update sequence has \textbf{\emph{load factor}} $1-\eps$ if, at every point in time, the total size of live objects, 
$\sum_{x\in X}\mu(x)$, is at most $ (1-\eps)M $.\footnote{A more relaxed definition of load factor, which appeared in some previous works such as \cite{Kuszmaul23}, allows the total size of live objects to be at most $\lfloor (1-\eps)M \rfloor + 1  $ (where $\eps>0$).
These two definitions are equivalent up to changing $\eps$ by a constant factor; see \cref{obs:makememoryfull}.}
The load factor is known to the allocator in advance.
Then, any no-move allocator (even allowing randomization) can only handle load factor $1-\eps \le O(\frac{1}{\log M})$, which tends to $0$ as $M$ grows \cite{Robson71,Robson74,LubyNO96,nicole}.

\paragraph*{Memory Reallocation.}
This situation has motivated the theoretical study of \emph{Memory \textbf{Re}allo\-cation} in various settings \cite{HallP04,SandersSS09,LimGL15,BenderFFFGspaa15,BenderFFFG15schedule,BenderFFFG17,Kuszmaul23,Farach-ColtonKS24}, where the goal is to support much higher load factor (ideally arbitrarily close to $1$), by allowing the allocator to move around existing objects in memory, while incurring a small reallocation overhead. 
We follow the formal setup of this problem considered by
Farach-Colton, Kuszmaul, Sheffield, and Westover \cite{Farach-ColtonKS24} (which was also studied earlier in \cite{Kuszmaul23,BenderFFFG17,NaorT01} with extra history-independent or cost-oblivious requirements; see \cref{sec:related}), defined as follows:

In the \emph{Memory Reallocation} problem, to handle an update, the allocator may move the existing objects, at an (unnormalized) \textbf{\emph{switching cost}} defined as the total size of moved objects (including the object being inserted/deleted). The \textbf{\emph{overhead factor}} (or \textbf{\emph{overhead}} for short) is the unnormalized switching cost divided by the size of the object being inserted/deleted. Formally, when inserting (or deleting) an object $x$, changing $X$ to $X' = X\sqcup \{x\}$ (or $X'=X\setminus \{x\}$), if the allocator changes the allocation map $\phi\colon X\to \{0,1,\dots,M-1\}$ to $\phi' \colon X' \to \{0,1,\dots,M-1\}$, then the overhead factor incurred for this update is
\[
  1 + \frac{1}{\mu(x)}\cdot \sum_{\substack{y\in X\cap X': \\ \phi(y)\neq \phi'(y)}} \mu(y).
\]

When the load factor is $1-\eps$, 
the following folklore deterministic allocator \cite{BenderFFFG17,Kuszmaul23} achieves an $O(\eps^{-1})$ overhead factor (remarkably, independent of the memory size $M$) for every update: to insert an object $x$, one simply finds an interval of size $O(\eps^{-1}\mu(x))$ with at least $\mu(x)$ empty memory slots (which must exist by an averaging argument), and reorganizes all objects contained in this interval to create a contiguous empty space that fits $x$. 

On the other hand, there has been no known lower bound that rules out constant overhead. 
This  raises the following main question:

\begin{quote}
   \centering
    \itshape  
    What is the best possible reallocation overhead as a function of the load factor $1-\eps$?
\end{quote}

To be more precise, when we say an allocator $\caA$ achieves overhead $f(\eps)$, we mean that for every $\eps>0$, there exists a family of allocators $\{\caA_M\}_{M\in \Z^+}$ such that for every $M\in \Z^+$, $\caA_M$ achieves overhead $f(\eps)$ on all input instances with $M$ memory slots and load factor $1-\eps$.

\paragraph*{Previous results.}
Recent works have made progress on this question via \emph{randomization}.
Formally, we say a randomized allocator achieves (worst-case) \textbf{\emph{expected overhead}} $f(\eps)$ (against an oblivious adversary), if for every update sequence  that is fixed in advance with load factor $1-\eps$, the overhead factor incurred for each update has expectation at most $f(\eps)$ over the randomness of the allocator. 

In the special case where every object has size at most $\eps^{4}M$, Kuszmaul \cite{Kuszmaul23} designed a randomized allocator with $O(\log \eps^{-1})$ expected overhead, exponentially better than the folklore allocator.
His allocator is based on a variant of uniform probing with strong history independence. 
For larger object sizes, however, he conjectured that the folklore $O(\eps^{-1})$ overhead was optimal.
Surprisingly, this conjecture was later disproven by Farach-Colton, Kuszmaul, Sheffield, and Westover \cite{Farach-ColtonKS24}, who designed a randomized allocator with expected overhead $O\big((\tfrac{1}{\eps})^{1/2}\,\polylog(\tfrac{1}{\eps})\big)$. 

 The allocators of \cite{Kuszmaul23} and \cite{Farach-ColtonKS24} additionally satisfy a \textbf{\emph{resizability}} property, namely that the live objects $X$ are always allocated to the prefix $\big [0, \frac{1}{1-\eps}\sum_{x\in X}\mu(x)\big )$ of the memory.
  Although no lower bound was known for general allocators, \cite{Farach-ColtonKS24} proved that resizable allocators must incur overhead $\Omega(\log \eps^{-1})$ (even if the entire update sequence is known in advance). 
  However, this still leaves an exponential gap from their $O\big((\tfrac{1}{\eps})^{1/2}\,\polylog(\tfrac{1}{\eps})\big)$ upper bound.

  Despite the lack of progress in the general case, several special classes of update sequences are known to admit logarithmic expected overhead. 
  Examples (in addition to the aforementioned tiny-object case~\cite{Kuszmaul23}) include:  update sequences where all objects have power-of-two sizes \cite{Kuszmaul23}, update sequences with only a constant number of distinct object sizes, all within a constant factor of each other (this was mentioned in \cite{Farach-ColtonKS24} as a corollary of their techniques), as well as a certain distribution of random update sequences \cite{Farach-ColtonKS24}.
Inspired by these successful examples, the authors of~\cite{Farach-ColtonKS24} speculated that the general case could possibly be tackled by a  structure-versus-randomness dichotomy, as is often featured in additive combinatorics. 
To add to this optimism, additive combinatorial techniques have led to advances in several algorithmic problems involving packing various-sized objects (such as Knapsack and Bin Packing; see \cref{sec:related}). Memory Reallocation appears to be another problem of this type, albeit with an additional dynamic flavor.

\subsection{Our results} 

\paragraph*{Upper bound.}
We present the first allocator with expected overhead factor $\polylog\!\left(\tfrac{1}{\eps}\right)$, exponentially improving the previous state-of-the-art result by
Farach-Colton, Kuszmaul, Sheffield, and Westover \cite{Farach-ColtonKS24}.
Like \cite{Farach-ColtonKS24} and \cite{Kuszmaul23}, our allocator is resizable. 
\begin{theorem}[\cref{sec:ds}]
   \label{thm:main}
The Memory Reallocation problem can be solved by a randomized resizable  allocator with worst-case expected overhead 
$O(\log^4 \epsilon^{-1}\cdot (\log \log \epsilon^{-1})^2)$ against an oblivious adversary.
\end{theorem}

As a confirmation of the conjecture in \cite{Farach-ColtonKS24}, \cref{thm:main} exploits the additive combinatorics of the object sizes, but the actual design of the allocator turns out to be less complicated than we expected.
The only combinatorial tool we need is a simple but beautiful theorem of Erd\H{o}s and S\'{a}rk\"{o}zy about subset sums \cite{erdossarkozy}, which they proved using the celebrated sunflower lemma pioneered by Erd\H{o}s and Rado \cite{erdosrado}. See \cref{subsec:overview} for a brief overview of the techniques.  %

By a simple transformation (see \cref{obs:makememoryfull}), \cref{thm:main} also implies an allocator with $\polylog(M)$ overhead, even when all the $M$ memory slots can be full.
\begin{corollary}
   \label{cor:main-full}
The Memory Reallocation problem (where the memory can be full) can be solved by a randomized resizable allocator with worst-case expected overhead $O(\log^4 M\cdot (\log \log M)^2)$ against an oblivious adversary.
\end{corollary}

\paragraph*{Lower bound.}
Our next result is a logarithmic lower bound for the overhead factor, which is also the first known super-constant lower bound that applies to general allocators without extra constraints.
Previously, a logarithmic lower bound was known for resizable allocators \cite{Farach-ColtonKS24},\footnote{An advantage of \cite{Farach-ColtonKS24}'s lower bound is that it also applies to offline resizable allocators, whereas ours only applies to online allocators.} and a near-logarithmic lower bound was known for strongly history-independent allocators \cite{Kuszmaul23}.

\begin{theorem}[\cref{sec:lb}]
   \label{thm:loglowerbound}
In the Memory Reallocation problem,  the worst-case expected overhead of any allocator (against an oblivious adversary) must be at least $\Omega (\log \epsilon^{-1})$.
\end{theorem}

Closing the near-quartic gap between our upper bound and lower bound is an interesting open question.

\paragraph*{High-probability guarantee?}
\cref{thm:main} only achieves worst-case \emph{expected} logarithmic overhead. On each update, our allocator may incur large overhead with non-negligible probability, due to periodically performing expensive rebuilds. 
This raises a natural question: can we improve \cref{thm:main} to achieve $\polylog \eps^{-1}$ overhead on every update \emph{with high probability} in $\eps^{-1}$, or even deterministically? 

Unfortunately, our next theorem implies that this is not possible. In the following, the \emph{squared overhead} incurred for an update is defined as the square of the overhead factor for that update.

\begin{theorem}[\cref{sec:lb}]
   \label{thm:mainsecondmoment}
In the Memory Reallocation problem,
  the worst-case expected squared overhead of any allocator (against an oblivious adversary) must be at least
$\Omega (\epsilon^{-1/7})$, even when all objects have size $\Theta(\eps^{3/7} M)$.
\end{theorem}

In the scenario of \cref{thm:mainsecondmoment}, the overhead factor of any update is always at most $O(\eps^{-3/7})$ (which corresponds to reorganizing the entire memory). Hence, if there is an allocator for load factor $1-\eps$ that achieves, for each update, an overhead at most $f$ with at least $1-\delta$ probability, then either $f \ge \Omega(\sqrt{\eps^{-1/7}}) = \Omega(\eps^{-1/14})$ or  $\delta \ge \Omega\big (\frac{\eps^{-1/7}}{(\eps^{-3/7})^2}\big ) =  \Omega(\eps^{5/7})$ must hold. 

The exponent $1/7$ in \cref{thm:mainsecondmoment} can be slightly improved by a more complicated refinement of our proof, which we omit in this paper. 
We leave it to future work to determine the tight bound.

Similarly to \cref{cor:main-full}, we can show variants of \cref{thm:loglowerbound,thm:mainsecondmoment} with lower bounds stated in terms of the number of memory slots $M$ (which may be full); see details in \cref{sec:lb}.

\subsection{Technical overview}

\label{subsec:overview}

\paragraph*{Upper bound.} 
We briefly describe the technical ingredients behind our proof of \cref{thm:main}.

Our starting point is a generic substitution strategy employed by  Farach-Colton, Kuszmaul, Sheffield, and Westover \cite{Farach-ColtonKS24} in their allocator beating the folklore $O(1/\eps)$ overhead.
Their allocator is resizable, i.e., the live objects $X$ are always allocated within the prefix $[0,\frac{1}{1-\eps}\sum_{x\in X}\mu(x))$ of memory.
For a resizable allocator, it is always safe to append any inserted object immediately after the rightmost live object in memory. The challenge is to restore resizability whenever a deletion happens and creates a gap in memory. 
If the deleted object $x$ is near the right end of the allocation, then it is cheap to shift leftward all the objects after $x$, filling in the gap. 
However, this becomes expensive if the deleted object $x$ is far from the right end.

A key idea of \cite{Farach-ColtonKS24} is a substitution strategy: find another object $y$ near the right end such that $\mu(x)-\mu(y)$ is non-negative but small, and move $y$ into $x$'s original place, nearly filling in the gap. The original place of $y$ then becomes a gap, but it is cheaper to fill it since it is closer to the right.
In order to ensure that such a good substitute object $y$ can always be found near the right end, \cite{Farach-ColtonKS24}'s allocator carefully sorts and organizes the objects in memory, and performs periodic rebuilds. Intuitively, if many updates have happened and consumed most of the available substitute objects on the right, then it performs an expensive rebuild to replenish the supply.

The allocator in \cite{Farach-ColtonKS24} guarantees that the substitute object $y$ and the deleted object $x$ have relative size difference $O(\eps^{1/2})$. While moderately small, this size difference is still non-negligible and accumulates over time, causing significant waste of memory.
This is the main reason why \cite{Farach-ColtonKS24} only achieves polynomial overhead instead of polylogarithmic.
On the positive side, \cite{Farach-ColtonKS24}'s substitution strategy with rebuilds can imply the following interesting corollary (which was noted in the conclusion section of their paper):
\begin{proposition}[based on \cite{Farach-ColtonKS24}]
   \label{prop:fewtype}
   If there are only $k$ distinct object sizes, all in $[\mu,2\mu]$, then an allocator can achieve worst-case expected overhead $O(k\log \frac{M}{k\mu})$. 
\end{proposition}
When very few distinct object sizes exist, \cref{prop:fewtype} offers a significantly improved overhead factor compared to the general case.
Roughly speaking, \cref{prop:fewtype} is possible because one can ensure that the substitute object $y$ always has exactly the same size as $x$.

Our strategy for proving \cref{thm:main} can be viewed as reducing from the general case to the special case of few distinct object sizes in \cref{prop:fewtype}. 
The key tool that enables this reduction is an additive combinatorial result of Erd\H{o}s and S\'{a}rk\"{o}zy~\cite{erdossarkozy}, which they proved using the sunflower lemma:  in any (multi)set $S$ of integers from $[n]$, one can find at least $\Omega(\frac{|S|}{\log^2 n})$ \emph{disjoint} subsets $\{B_i\}_i$, each of cardinality $|B_i|\le O(\log n)$, which all have equal sum $\sum_{a\in B_i}a$.\footnote{The actual theorem in \cite{erdossarkozy} was not this statement, but rather its direct corollary  that the subset sums of $S$ contain an arithmetic progression of length $\Omega(\frac{|S|}{\log^2 n})$.
The latter statement is of more historical interest in the literature, but less useful for us.
}

When the object sizes are within a constant factor of each other,
we can appropriately round each object size up by a factor of at most $1+\eps$ (which is allowable when the load factor is $1-O(\eps)$), and then iteratively apply the above lemma with $O(\eps^{-1})$ in place of~$n$.
This  partitions the objects into \emph{bundles} of $O(\log \eps^{-1})$ objects each, with at most $\polylog(\eps^{-1})$ distinct bundle sizes.
We allocate each bundle contiguously in memory. 
In this way, our situation resembles the special case of few distinct object sizes as in \cref{prop:fewtype}.
The difference is that we operate on bundles rather than individual objects; this increases the overhead factor by the number of objects in each bundle, which is only logarithmic.

Our proof of \cref{thm:main} proceeds not by a black-box reduction to \cref{prop:fewtype}, but rather by modifying its proof and handling additional implementation details related to bundling. For example, we need to unbundle a bundle when one of its objects is deleted, and bundle new objects as they are inserted.
During each rebuild operation, we will rebundle the affected objects, so that the number of distinct bundle sizes remains bounded by $\polylog(\eps^{-1})$. (We also need to lift the bounded-ratio assumption in \cref{prop:fewtype}; this was already addressed by \cite{Farach-ColtonKS24}, and we omit the details in this overview.)

Since \cite{Farach-ColtonKS24} did not explicitly prove \cref{prop:fewtype}, we include a proof below for convenience. This also serves as a warm-up for our main proof of \cref{thm:main}.

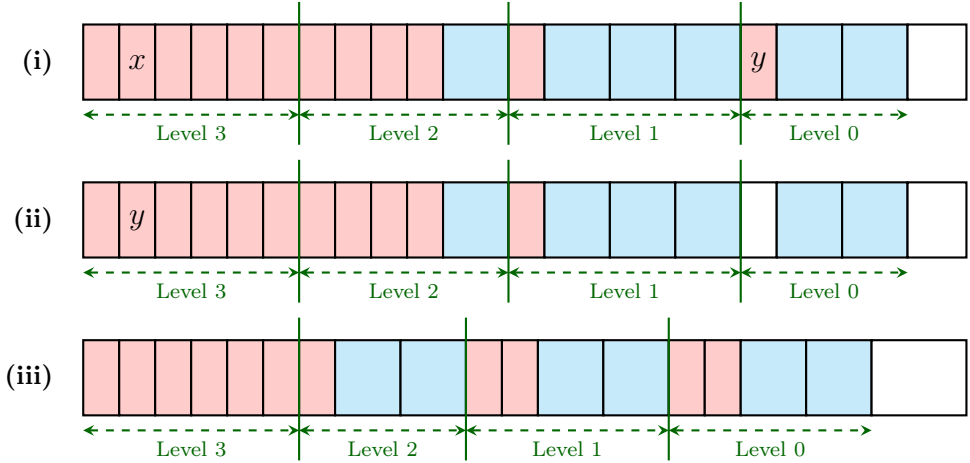
\begin{figure}[htbp]
    \centering
    \begin{tikzpicture}[
    x=0.0087cm, y=1cm, 
    block/.style={
        draw=black, 
        thick
    },
    label/.style={
        anchor=center,
        font=\large
    },
    row label/.style={
        anchor=east,
        font=\small\bfseries
    },
    level arrow/.style={
        <->,
        dashed,
        green!40!black,
        thick,
        >=stealth
    },
    level sep/.style={
        green!40!black,
        thick
    },
    level text/.style={
        midway, 
        below, 
        font=\scriptsize, 
        text=green!40!black
    }
]

    \def\sb{100}
    \def\sa{55}
    \def\sc{115}
    \def\sf{90}
    \def\sfp{145}
    
    \def\strA{a}
    \def\strB{b}

    \def\yA{4.2} 
    \node[row label] at (-30, \yA+0.5) {(i)};

    \def\rowOne{
        \sa/a, \sa/a, \sa/a, \sa/a, \sa/a, \sa/a,
        \sa/a, \sa/a, \sa/a, \sa/a, \sb/b,
        \sa/a, \sb/b, \sb/b, \sb/b,
        \sa/a, \sb/b, \sb/b, 
        \sf/x
    }

    \def\xPos{0}
    \coordinate (R1_start) at (0, \yA);

    \foreach \width/\lab [count=\i] in \rowOne {
        \def\currFill{white}
        \ifx\lab\strA \def\currFill{red!20} \fi
        \ifx\lab\strB \def\currFill{cyan!20} \fi
        \ifnum\i>18 \def\currFill{white} \fi

        \draw[block, fill=\currFill] (\xPos, \yA) rectangle ++(\width, 1);
        
        \ifnum\i=2 \node[label] at (\xPos + 0.5*\width, \yA + 0.5) {$x$}; \fi
        \ifnum\i=16 \node[label] at (\xPos + 0.5*\width, \yA + 0.5) {$y$}; \fi
        
        \pgfmathsetmacro{\xPos}{\xPos + \width}
        \xdef\xPos{\xPos}

        \ifnum\i=6  \coordinate (R1_L3) at (\xPos, \yA); \fi 
        \ifnum\i=11 \coordinate (R1_L2) at (\xPos, \yA); \fi 
        \ifnum\i=15 \coordinate (R1_L1) at (\xPos, \yA); \fi 
        \ifnum\i=18 \coordinate (R1_L0) at (\xPos, \yA); \fi 
    }

    \draw[level arrow] (R1_start |- 0,\yA-0.2) -- (R1_L3 |- 0,\yA-0.2) node[level text] {Level 3};
    \draw[level arrow] (R1_L3 |- 0,\yA-0.2)    -- (R1_L2 |- 0,\yA-0.2) node[level text] {Level 2};
    \draw[level arrow] (R1_L2 |- 0,\yA-0.2)    -- (R1_L1 |- 0,\yA-0.2) node[level text] {Level 1};
    \draw[level arrow] (R1_L1 |- 0,\yA-0.2)    -- (R1_L0 |- 0,\yA-0.2) node[level text] {Level 0};

    \draw[level sep] (R1_L3 |- 0,\yA-0.6) -- (R1_L3 |- 0,\yA+1.3);
    \draw[level sep] (R1_L2 |- 0,\yA-0.6) -- (R1_L2 |- 0,\yA+1.3);
    \draw[level sep] (R1_L1 |- 0,\yA-0.6) -- (R1_L1 |- 0,\yA+1.3);

    \def\yB{2.1} 
    \node[row label] at (-30, \yB+0.5) {(ii)};

    \def\xPos{0}
    \coordinate (R2_start) at (0, \yB);

    \foreach \width/\lab [count=\i] in \rowOne {
        \def\currFill{white}
        \ifx\lab\strA \def\currFill{red!20} \fi
        \ifx\lab\strB \def\currFill{cyan!20} \fi
        \ifnum\i=16 \def\currFill{white} \fi
        \ifnum\i>18 \def\currFill{white} \fi

        \draw[block, fill=\currFill] (\xPos, \yB) rectangle ++(\width, 1);
        
        \ifnum\i=2 \node[label] at (\xPos + 0.5*\width, \yB + 0.5) {$y$}; \fi
        
        \pgfmathsetmacro{\xPos}{\xPos + \width}
        \xdef\xPos{\xPos}

        \ifnum\i=6  \coordinate (R2_L3) at (\xPos, \yB); \fi
        \ifnum\i=11 \coordinate (R2_L2) at (\xPos, \yB); \fi
        \ifnum\i=15 \coordinate (R2_L1) at (\xPos, \yB); \fi
        \ifnum\i=18 \coordinate (R2_L0) at (\xPos, \yB); \fi
    }

    \draw[level arrow] (R2_start |- 0,\yB-0.2) -- (R2_L3 |- 0,\yB-0.2) node[level text] {Level 3};
    \draw[level arrow] (R2_L3 |- 0,\yB-0.2)    -- (R2_L2 |- 0,\yB-0.2) node[level text] {Level 2};
    \draw[level arrow] (R2_L2 |- 0,\yB-0.2)    -- (R2_L1 |- 0,\yB-0.2) node[level text] {Level 1};
    \draw[level arrow] (R2_L1 |- 0,\yB-0.2)    -- (R2_L0 |- 0,\yB-0.2) node[level text] {Level 0};

    \draw[level sep] (R2_L3 |- 0,\yB-0.6) -- (R2_L3 |- 0,\yB+1.3);
    \draw[level sep] (R2_L2 |- 0,\yB-0.6) -- (R2_L2 |- 0,\yB+1.3);
    \draw[level sep] (R2_L1 |- 0,\yB-0.6) -- (R2_L1 |- 0,\yB+1.3);

    \def\yC{0}
    \node[row label] at (-30, \yC+0.5) {(iii)};

    \def\rowThree{
        \sa/a, \sa/a, \sa/a, \sa/a, \sa/a, \sa/a,
        \sa/a, \sb/b, \sb/b,
        \sa/a, \sa/a, \sb/b, \sb/b, 
        \sa/a, \sa/a, \sb/b, \sb/b, 
        \sfp/x
    }

    \def\xPos{0}
    \coordinate (R3_start) at (0, \yC);

    \foreach \width/\lab [count=\i] in \rowThree {
        \def\currFill{white}
        \ifx\lab\strA \def\currFill{red!20} \fi
        \ifx\lab\strB \def\currFill{cyan!20} \fi
        \ifnum\i>17 \def\currFill{white} \fi

        \draw[block, fill=\currFill] (\xPos, \yC) rectangle ++(\width, 1);
        
        \pgfmathsetmacro{\xPos}{\xPos + \width}
        \xdef\xPos{\xPos}

        \ifnum\i=6  \coordinate (R3_L3) at (\xPos, \yC); \fi 
        \ifnum\i=9  \coordinate (R3_L2) at (\xPos, \yC); \fi 
        \ifnum\i=13 \coordinate (R3_L1) at (\xPos, \yC); \fi 
        \ifnum\i=17 \coordinate (R3_L0) at (\xPos, \yC); \fi 
    }

    \draw[level arrow] (R3_start |- 0,\yC-0.2) -- (R3_L3 |- 0,\yC-0.2) node[level text] {Level 3};
    \draw[level arrow] (R3_L3 |- 0,\yC-0.2)    -- (R3_L2 |- 0,\yC-0.2) node[level text] {Level 2};
    \draw[level arrow] (R3_L2 |- 0,\yC-0.2)    -- (R3_L1 |- 0,\yC-0.2) node[level text] {Level 1};
    \draw[level arrow] (R3_L1 |- 0,\yC-0.2)    -- (R3_L0 |- 0,\yC-0.2) node[level text] {Level 0};

    \draw[level sep] (R3_L3 |- 0,\yC-0.6) -- (R3_L3 |- 0,\yC+1.3);
    \draw[level sep] (R3_L2 |- 0,\yC-0.6) -- (R3_L2 |- 0,\yC+1.3);
    \draw[level sep] (R3_L1 |- 0,\yC-0.6) -- (R3_L1 |- 0,\yC+1.3);

    \end{tikzpicture}
    \caption{An illustration for the proof of \cref{prop:fewtype} with two types of object sizes. 
    (i)~Initial state at $c=3$, where object $x$ is to be deleted. 
    We find another object $y$ of the same size in level $0$. 
    (ii)~Swap $x$ and $y$ before deleting $x$. 
    (iii)~Increase the counter to $c=4$, and rebuild levels $2, 1$, and $0$.}
    \label{fig:upperbound}
\end{figure}

\begin{proof}[Proof of \cref{prop:fewtype} (based on \cite{Farach-ColtonKS24})]
We maintain the \emph{prefix property}, namely that all live objects occupy a prefix of the memory without any gaps in between.
The objects are partitioned into $\ell \coloneqq \lceil \log (\frac{M}{k\mu})\rceil$ \emph{levels} numbered $\ell-1,\dots,1,0$ from left to right, each occupying a contiguous interval of memory.\footnote{This is slightly different from \cite{Farach-ColtonKS24}, who defined levels as nested subsets instead of disjoint subsets.}

Intuitively, levels further to the left can have larger total size and are less frequently rebuilt. Formally,
we maintain the following two invariants, where the counter $c$ is initialized to a random integer and incremented by one per update (i.e., insertion or deletion).
\begin{itemize}
    \item For every integer $j\in [1,\ell]$, the total number of objects in levels $\{j-1,\dots,1,0\}$ is at most $k\cdot 2^{j} + (c\bmod 2^j) < (k+1)2^j$.
    \item For every object size $\tau$ and integer $j\in [1,\ell-1]$, if levels $\{\ell-1,\dots,j+1,j\}$ contain at least one  size-$\tau$ object, then the total number of size-$\tau$ objects in levels $\{j-1,\dots,1,0\}$ is at least $2^{j} - (c\bmod 2^j) \ge 1$.
\end{itemize}
Clearly, the two invariants hold at the beginning with empty memory.

To insert an object, we add it to level $0$ immediately after the current rightmost object.

To delete a size-$\tau$ object $x$, if $x$ is not already in level $0$, we swap $x$ with another size-$\tau$ object~$y$ from level $0$, which must exist due to the second invariant. 
See \cref{fig:upperbound}.
Hence, the deletion only creates a gap in level $0$, which will be filled during the rebuild operation described below.

After every update, we increment the counter $c$ and then  perform a rebuild as follows. 
Pick the largest integer $j^*\le \ell-1$ such that $2^{j^*}$ divides $c$. 
Then, we collect all objects in levels $\{j^*,\dots,1,0\}$, and repartition them by the following greedy rule: Define $(n_{0},n_{1},\dots,n_{j^*-1}) \coloneqq (2,2,4,8,16,32,\dots)$ and $n_{j^*}\coloneqq +\infty$.
For each object size $\tau$, we put $n_{0}$ size-$\tau$ objects in level $0$, then put $n_{1}$ of the remaining size-$\tau$ objects in level $1$, and so on, until the number of remaining size-$\tau$ objects is smaller than $n_j$ for the current level $j$, at which point we finish by putting all of them in level $j$.

To prove the first invariant for $j\in [1,\ell-1]$ (the $j=\ell$ case vacuously holds by our definition of $\ell$), consider the most recent rebuild that involves level $j$, which happened $c\bmod 2^j$ updates before (if no such rebuild happened before, the proof is similar).
Right after this rebuild, the number of objects in levels $\{j-1,\dots,1,0\}$ is at most $n_0+n_1+\dots +n_{j-1}=2^{j}$ for each object size $\tau$, and hence at most $k\cdot 2^j$ in total. Each subsequent update increases this count by at most one, so the first invariant holds.
The second invariant can be proved similarly (we omit the details here).

For each update, a rebuild of levels $\{j,\dots,1,0\}$ is triggered with probability $O(2^{-j})$, incurring a switching cost bounded by their total size, which is $O(k\cdot 2^{j}\cdot \mu)$ by the first invariant. 
Hence, the expected overhead factor is $\sum_{j=0}^{\ell-1}O(2^{-j})\cdot O(k\cdot 2^{j}\cdot \mu)/\mu = O(k\ell) = O(k\log(\frac{M}{k\mu}))$ as claimed.
\end{proof}

\paragraph*{Lower bounds (general strategy).}
We begin with a minor technical simplification for all our lower bound proofs: by a rounding argument (see \cref{obs:makememoryfull}), it suffices to prove overhead lower bounds in the scenario where the number of slots is $M = \Theta(1/\eps)$ and the memory may be full. See details in \cref{sec:lb}.

Now we discuss the general proof ideas for \cref{thm:loglowerbound} and \cref{thm:mainsecondmoment}.
Recall the allocator in \cref{thm:main} exploits the ``additive coincidences'' of object sizes, namely different small subsets of equal sum, to perform cheap substitutions.
The proofs of lower bounds proceed in the opposite direction: we design hard instances that lack additive coincidences, so that cheap substitutions are impossible.
This idea was already reflected in the previous lower bound by \cite{Farach-ColtonKS24} against resizable allocators, and we will further build on this idea. 

We will frequently use the following simple but crucial notion implicitly introduced by \cite{Farach-ColtonKS24}: 
between two full memory states $\phi,\phi'$,  a \emph{maximal changed interval} is an inclusion-maximal interval of memory such that no object in $\phi$ intersecting this interval retains its location in $\phi'$. 
Let $\Delta(\phi,\phi')$ denote the total length of all the maximal changed intervals.
Then, transforming $\phi$ into $\phi'$ requires a total switching cost at least $\Delta(\phi,\phi')$.
See \cref{fig:examplephi} for an illustration and \cref{subsec:lbdefs} for formal definitions.

A maximal changed interval either permutes the objects inside, or substitutes them by a different set of objects with the same total size. 
At a high level, by designing object sizes with few additive coincidences, we have control of the possible substitutions that might occur. This helps us show limitations of low-overhead allocators, which can only produce short maximal changed intervals.

\begin{figure}[htbp]
   \centering
\begin{tikzpicture}[x=1cm, y=1cm]

    \tikzset{
        cell/.style={draw=gray!50, thin},
        block/.style={fill=lightgray!50, draw=black, thick, rounded corners=2pt},
        highlight/.style={fill=gray!70, draw=black, thick, rounded corners=2pt},
        label/.style={font=\sffamily\bfseries}
    }

    \def\yA{1.5} %
    
    \node[anchor=east] at (-0.2, \yA+0.5) {$\phi$};

    \foreach \x in {1,...,13} {
        \draw[black!70, thin] (\x, \yA) -- (\x, \yA+1);
    }

    \draw[black, thick] (0, \yA) rectangle (14, \yA+1); %

    \draw[block] (0.1, \yA+0.1) rectangle (1.9, \yA+0.9);
    \node at (1, \yA+0.5) {2};

    \draw[highlight] (2.1, \yA+0.1) rectangle (6.9, \yA+0.9);
    \node at (4.5, \yA+0.5) {5};

    \draw[block] (7.1, \yA+0.1) rectangle (9.9, \yA+0.9);
    \node at (8.5, \yA+0.5) {3};

    \draw[block] (10.1, \yA+0.1) rectangle (10.9, \yA+0.9);
    \node at (10.5, \yA+0.5) {1};

    \draw[highlight] (11.1, \yA+0.1) rectangle (13.9, \yA+0.9);
    \node at (12.5, \yA+0.5) {3};

    \def\yB{0} %

    \node[anchor=east] at (-0.2, \yB+0.5) {$\phi'$};

    \foreach \x in {1,...,13} {
        \draw[black!70, thin] (\x, \yB) -- (\x, \yB+1);
    }

    \draw[black, thick] (0, \yB) rectangle (14, \yB+1); %

    \draw[block] (0.1, \yB+0.1) rectangle (0.9, \yB+0.9);
    \node at (0.5, \yB+0.5) {1};

    \draw[block] (1.1, \yB+0.1) rectangle (1.9, \yB+0.9);
    \node at (1.5, \yB+0.5) {1};

    \draw[highlight] (2.1, \yB+0.1) rectangle (6.9, \yB+0.9);
    \node at (4.5, \yB+0.5) {5};

    \draw[block] (7.1, \yB+0.1) rectangle (7.9, \yB+0.9);
    \node at (7.5, \yB+0.5) {1};

    \draw[block] (8.1, \yB+0.1) rectangle (10.9, \yB+0.9);
    \node at (9.5, \yB+0.5) {3};

    \draw[highlight] (11.1, \yB+0.1) rectangle (13.9, \yB+0.9);
    \node at (12.5, \yB+0.5) {3};

    \draw[|-|, line width=1.5pt, dashed, blue] (0, 1.25) -- (2, 1.25) ;
    \draw[|-|, line width=1.5pt, dashed, blue] (7, 1.25) -- (11, 1.25) ;

    \foreach \x in {0,...,13} {
        \node[above] at (\x + 0.5, \yA+1) {\x};
    }

\end{tikzpicture}

   \caption{Visualization of two full memory states $\phi,\phi'$
 with $M=14$ slots. 
   The unchanged objects in $\phi$ and $\phi'$ (highlighted in dark gray) have sizes $5$ and $3$.
   The maximal changed intervals between them (highlighted in dashed blue) are
   $[0,2)$ and $[7,11)$.
\\
To transform $\phi$ to $\phi'$, the
allocator must  incur a total switching cost of at least $\Delta(\phi,\phi')=2+3+1=6$. 
}
   \label{fig:examplephi}
\end{figure}
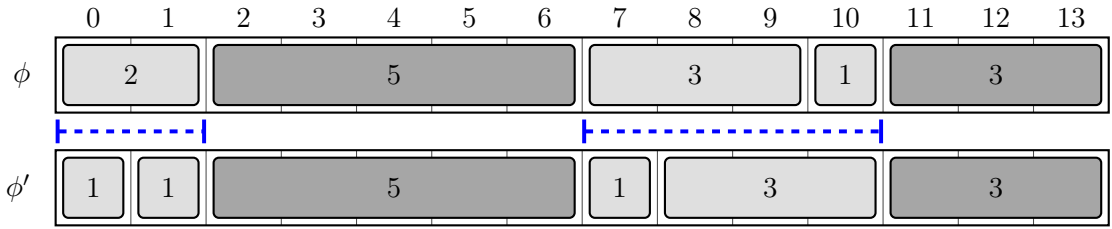

\paragraph*{Logarithmic lower bound.}
Now we sketch the proof of 
\cref{thm:loglowerbound}, which is relatively simple.
Recall that it suffices to prove a worst-case expected $\Omega(\log M)$ overhead lower bound when all $M$ memory slots can be full.
We construct a set $S=\{s_1,s_2,\dots,s_k\} \subset \Z^+$ such that for every pair of distinct $a,b\in [k]$, there exist $s_a',s_b'\in \Z^+$ with $s_a'+s_b'=s_a+s_b$ so that the following property holds:
\begin{itemize}
   \item Let  $S' = (S \setminus \{s_a,s_b\}) \sqcup \{s'_a,s'_b\}$.  If $X\subseteq S$ and $X'\subseteq S'$ satisfy $\sum_{x\in X}x = \sum_{x'\in X'}x'$ and $s_a \in X$, then $s_b\in X$ must hold as well.
\end{itemize}
One can verify that the following construction works and ensures every integer is bounded by $O(2^{2k})$: $s_i \coloneqq  2^{i+k} + 2^i$ for all $i\in [k]$, $s'_a \coloneqq 2^{a+k}+2^b$ and $s'_b \coloneqq 2^{b+k}+2^a$.
We further modify this construction to make every integer in $\Theta(2^{2k})$ by padding (see details in \cref{sec:loglb}).
Define the total number of memory slots to be $M\coloneqq s_1+\dots +s_k$ (hence, $k = \Theta(\log M)$ and $s_i \in \Theta(M/\log M)$).
In the hard instance, we first insert $k$ objects of sizes $s_1,s_2,\dots,s_k$ to completely fill the memory, and denote the current memory state by $\phi$.
Then, we pick two distinct indices $a,b\in [k]$ uniformly at random, delete the objects of sizes $s_a,s_b$ and insert objects of sizes $s'_a,s'_b$ as defined above.
Denote the current memory state by $\phi'$.
Now, consider the maximal changed interval $[L,R)$ between $\phi,\phi'$ that contains the size-$s_a$ object of $\phi$. Let $X$ denote the set of sizes of objects in $\phi$ contained in $[L,R)$, so $s_a\in X$. Then, the property above implies that $s_b\in X$ as well. Therefore, $R-L$ is at least the distance between the size-$s_a$ and size-$s_b$ objects in $\phi$, which has expectation $\Omega(M)$ over the randomness of $a$ and $b$. Thus, the total expected switching cost for the two insertions and two deletions is at least $\Ex[\Delta(\phi,\phi')] \ge \Ex[R-L] \ge \Omega(M) $, so the worst-case expected overhead is $\Omega(M)/\Theta(M/\log M) =\Omega(\log M) $ as claimed.

\paragraph*{High-probability lower bound.}
Our proof of \cref{thm:mainsecondmoment} is more involved. 
As a warm-up, it is instructive to see a proof sketch of the following worst-case lower bound against deterministic allocators with the prefix property (which is more or less the same as resizability):
\begin{proposition}
   \label{prop:warmup}
   Any deterministic allocator that satisfies the prefix property, namely that all live objects must occupy a prefix of the memory without any gaps in between, must incur worst-case overhead $\Omega(M^{1/4})$ when there are $M$ memory slots which can be full.
\end{proposition}

Compared to \cref{thm:mainsecondmoment}, here we are restricting to deterministic allocators to avoid probabilistic technicalities, but the main simplification of the proof comes from the prefix property assumption.

\begin{proof}[Proof of \cref{prop:warmup}]
 We pick two integers $a,b\in [\mu,2\mu]$ with $\mu = \Theta(\sqrt{M})$, so that all integers $ia+jb$ where $i,j\in \Z \cap [0,M/\mu]$ are distinct.
 For example, $a \coloneqq \mu\coloneqq  \lceil \sqrt{M}\rceil$ and $b\coloneqq a+1$ satisfy this property.
In our hard update sequence, we first insert $\lfloor M/b\rfloor$ objects of size $b$. 
Then, we gradually replace all of them by size-$a$ objects as follows:  in every iteration, we delete one size-$b$ object, and insert $O(1)$ size-$a$ objects until the total size of  live objects is in $(M-a,M]$. 
 (This  is almost the same as the update sequence in \cite{Farach-ColtonKS24}'s proof of the amortized logarithmic lower bound against resizable allocators.)
Denote the memory state right after the $i$-th iteration ($i\le \lfloor M/b\rfloor$) by $\phi_i$.
For a size-$a$ object at location $\lambda$, define its \emph{potential} to be $M-\lambda\ge 0$; the potential of a memory state is simply the total potential of all the size-$a$ objects in it.
At the beginning, the total potential is zero.
In the end, among all the $\lfloor M/a\rfloor $ size-$a$ objects, at least half of them have potential $\Omega(M)$, with total potential $\Omega(M^2/a) = \Omega(M^{3/2})$.

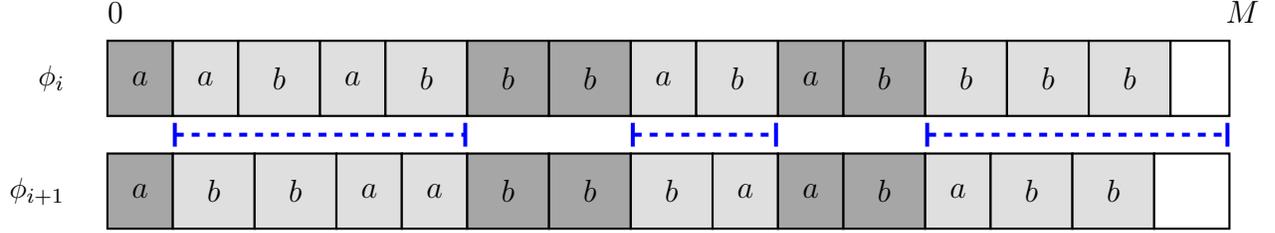
\begin{figure}[htbp]
   \centering
 
   \begin{tikzpicture}[
    x=0.0087cm, y=1cm, 
    block/.style={
        draw=black, 
        thick
    },
    label/.style={
        anchor=center,
        font=\large
    }
]

    \def\sa{100}
    \def\sb{125}
    \def\sc{115}
    \def\sf{90}
    \def\sfp{190}

    \def\yA{1.5} 
    \node[anchor=east] at (-50, \yA+0.5) {$\phi_i$};

    \def\rowOne{
        \sa/a, \sa/a, \sb/b, \sa/a, \sb/b, \sb/b, \sb/b,  \sa/a, \sb/b, \sa/a, \sb/b, \sb/b, \sb/b,\sb/b,\sf/
    }

    \def\xPos{0}
    \foreach \width/\lab [count=\i] in \rowOne {
        
        \def\currFill{gray!70} 
        \ifnum\i>1 
            \ifnum\i<6 
                \def\currFill{lightgray!50} 
            \fi 
        \fi

        \ifnum\i>7
            \ifnum\i<10
                \def\currFill{lightgray!50} 
            \fi
        \fi

        \ifnum\i>11
                \def\currFill{lightgray!50} 
        \fi

        \ifnum\i>14
                \def\currFill{white} 
        \fi

        \draw[block, fill=\currFill] (\xPos, \yA) rectangle ++(\width, 1);
        \node[label] at (\xPos + 0.5*\width, \yA + 0.5) {$\lab$};
        
        \ifnum\i=1 \coordinate (R1Start) at (\xPos + \width, 1.25); \fi
        \ifnum\i=5 \coordinate (R1End) at (\xPos + \width, 1.25); \fi
        
        \pgfmathsetmacro{\xPos}{\xPos + \width}
        \xdef\xPos{\xPos}
    }
    \node[above, font=\large] at (0, \yA + 1.1) {$\,\,\,0$};
    \node[above, font=\large] at (\xPos, \yA + 1.1) {$\,\,\,\,\,M$};

    \def\yB{0}
    \node[anchor=east] at (-50, \yB+0.5) {$\phi_{i+1}$};

    \def\rowTwo{
       \sa/a, \sb/b, \sb/b, \sa/a, \sa/a, \sb/b, \sb/b, \sb/b,  \sa/a, \sa/a, \sb/b,\sa/a,\sb/b,\sb/b,\sc/
    }

    \def\xPos{0}
    \foreach \width/\lab [count=\i] in \rowTwo {
        
          \def\currFill{gray!70}
        \ifnum\i>7
            \ifnum\i<10
                \def\currFill{lightgray!50} 
            \fi
        \fi

        \ifnum\i>1
            \ifnum\i<6
                \def\currFill{lightgray!50} 
            \fi
        \fi
        \ifnum\i>11
                \def\currFill{lightgray!50} 
        \fi

        \ifnum\i>14
                \def\currFill{white} 
        \fi

        \draw[block, fill=\currFill] (\xPos, \yB) rectangle ++(\width, 1);
        \node[label] at (\xPos + 0.5*\width, \yB + 0.5) {$\lab$};
        
        \ifnum\i=7 \coordinate (R2Start) at (\xPos + \width, 1.25); \fi
        \ifnum\i=9 \coordinate (R2End) at (\xPos + \width, 1.25); \fi

        \ifnum\i=11 \coordinate (R3Start) at (\xPos + \width, 1.25); \fi

        \ifnum\i=15 \coordinate (R3End) at (\xPos + \width, 1.25); \fi
        
        \pgfmathsetmacro{\xPos}{\xPos + \width}
        \xdef\xPos{\xPos}
    }

    \draw[|-|, line width=1.5pt, dashed, blue] (R1Start |- 0, 1.25) -- (R1End |- 0, 1.25);
    \draw[|-|, line width=1.5pt, dashed, blue] (R2Start |- 0, 1.25) -- (R2End |- 0, 1.25);
    \draw[|-|, line width=1.5pt, dashed, blue] (R3Start |- 0, 1.25) -- (R3End |- 0, 1.25);

\end{tikzpicture}
\caption{An illustration for the proof of \cref{prop:warmup}. Between memory states $\phi_i,\phi_{i+1}$ satisfying the prefix property, every maximal changed interval $[L_k,R_k)$ (marked in dashed blue), except for the rightmost one, permutes the objects inside. } 
\label{fig:warm}
\end{figure}

Between two adjacent memory states $\phi_i,\phi_{i+1}$, consider all the maximal changed intervals $[L_k,R_k)$, with the exception that the interval $[L_f,R_f)$ containing the rightmost object of $\phi_{i+1}$ is redefined with right boundary $R_f\coloneqq M$. 
See an illustration in \cref{fig:warm}.
Since the rightmost gap in $\phi_{i+1}$ is smaller than $a$, the total switching cost to transform $\phi_i$ into $\phi_{i+1}$ is at least $\sum_{k \neq f}(R_k-L_k) + (R_f-L_f-a)$. Assuming the allocator incurs worst-case switching cost $S$ per update, this implies $\sum_{k}(R_k-L_k)  \le  O(S) + a$.

Every $[L_k,R_k)$, except for the rightmost one, is completely filled by objects, and hence must contain the same number of size-$a$ and size-$b$ objects in both states, due to the property of $a$ and $b$ stated at the beginning.
Now we analyze the total potential difference between $\phi_i$ and $\phi_{i+1}$: 
\begin{itemize}
   \item Inside every $[L_k,R_k)$ ($k\neq f$), we can perfectly pair up the $a$-objects in $\phi_i$ with those in $\phi_{i+1}$,  forming at most $(R_k-L_k)/a$ pairs. Each pair contributes potential difference less than $R_k-L_k$. Hence, the total contribution to the potential difference is $\le (R_k-L_k)^2/a$.
   \item Inside the rightmost maximal changed interval $[L_f,R_f)$ in $\phi_i$ (and in $\phi_{i+1}$),
   there are  at most $(R_f-L_f)/a$ size-$a$ objects, each with potential at most $M-L_f = R_f-L_f$.
   Hence, the total contribution to the potential difference is also $\le (R_f-L_f)^2/a$.
\end{itemize}
Summing up, the total potential difference between $\phi_i,\phi_{i+1}$ is at most \[\sum_{k}(R_k-L_k)^2/a \le \big (\sum_k (R_k-L_k)\big )^2/a \le (O(S)+a)^2/a = O(S^2/a + a).\]
Summing up over all $i\le \lfloor M/b\rfloor $, we obtain that the potential difference between the initial and final states is $O(\frac{M}{b}(\frac{S^2}{a}+a)) =  O(S^2 + M)$.
Since we showed earlier that this potential difference is $\Omega(M^{3/2})$, this implies $S\ge \Omega(M^{3/4})$ for sufficiently large $M$.
Hence, the worst-case overhead factor is at least $\Omega(S/\mu) = \Omega(M^{1/4})$ as claimed.
\end{proof}

Relaxing the prefix requirement in  \cref{prop:warmup} poses more technical challenges.
In particular, the update sequence in the proof of \cref{prop:warmup} becomes easy: an allocator could always store size-$a$ objects in a prefix of the memory, and size-$b$ objects in a suffix, leaving a gap somewhere in the middle. Each update can be handled with constant overhead, by performing insertions and deletions near the gap. 

We now sketch how to modify the proof of \cref{prop:warmup} to obtain an $M^{\Omega(1)}$ worst-case bound for general deterministic allocators (without details on the calculation of parameters).
We pick object sizes $a,b\in  \Theta(\mu)$ with the same roles as before, and another object size $c\in \Theta(\mu)$ whose role will be explained later. 
In our construction, we ensure the total size of size-$a$, size-$b$, and size-$c$ objects stays in $M - \Theta(\mu)$, and we always insert another special object $f$ (which we call the ``finger object'') so that all $M$ memory slots are fully occupied. 
By selecting $a,b,c$ to avoid additive coincidences, we can achieve the following crucial property (formally stated in \cref{lem:finger}):
\begin{itemize}
   \item 
   Between two full memory states $\phi$ and $\phi'$, if a maximal changed interval does not contain the finger object of either state, then it can only permute the objects within it.
\end{itemize}

\begin{figure}[htbp]
   \centering
 
   \begin{tikzpicture}[
    x=0.0087cm, y=1cm, 
    block/.style={
        draw=black, 
        thick
    },
    label/.style={
        anchor=center,
        font=\large
    }
]

    \def\sa{100}
    \def\sb{125}
    \def\sc{145}
    \def\sf{209}
    \def\sfp{190}

    \def\yA{0} 
    \node[anchor=east] at (-50, \yA+0.5) {$\phi'$};

    \def\rowOne{
        \sa/a, \sa/a, \sc/c, \sa/a, \sb/b, \sb/b, \sc/c, \sf/f', \sb/b, \sb/b, \sc/c, \sb/b, \sc/c
    }

    \def\xPos{0}
    \foreach \width/\lab [count=\i] in \rowOne {
        
        \def\currFill{gray!70} 
        \ifnum\i>1 
            \ifnum\i<6 
                \def\currFill{lightgray!50} 
            \fi 
        \fi

        \ifnum\i>7
            \ifnum\i<11
                \def\currFill{lightgray!50} 
            \fi
        \fi

        \ifnum\i=8
                \def\currFill{lightgray!10} 
        \fi
        \draw[block, fill=\currFill] (\xPos, \yA) rectangle ++(\width, 1);
        \node[label] at (\xPos + 0.5*\width, \yA + 0.5) {$\lab$};
        
        \ifnum\i=1 \coordinate (R1Start) at (\xPos + \width, 1.25); \fi
        \ifnum\i=5 \coordinate (R1End) at (\xPos + \width, 1.25); \fi
        
        \pgfmathsetmacro{\xPos}{\xPos + \width}
        \xdef\xPos{\xPos}
    }

    \def\yB{1.5}
    \node[anchor=east] at (-50, \yB+0.5) {$\phi$};

    \def\rowTwo{
       \sa/a, \sb/b, \sc/c, \sa/a, \sa/a, \sb/b, \sc/c, \sb/b, \sfp/f, \sc/c, \sc/c, \sb/b,\sc/c
    }

    \def\xPos{0}
    \foreach \width/\lab [count=\i] in \rowTwo {
        
          \def\currFill{gray!70}
        \ifnum\i>7
            \ifnum\i<11
                \def\currFill{lightgray!50} 
            \fi
        \fi

        \ifnum\i=9
                \def\currFill{lightgray!10} 
        \fi

        \ifnum\i>1
            \ifnum\i<6
                \def\currFill{lightgray!50} 
            \fi
        \fi

        \draw[block, fill=\currFill] (\xPos, \yB) rectangle ++(\width, 1);
        \node[label] at (\xPos + 0.5*\width, \yB + 0.5) {$\lab$};
        
        \ifnum\i=7 \coordinate (R2Start) at (\xPos + \width, 1.25); \fi
        \ifnum\i=10 \coordinate (R2End) at (\xPos + \width, 1.25); \fi
        
        \pgfmathsetmacro{\xPos}{\xPos + \width}
        \xdef\xPos{\xPos}
    }

    \draw[|-|, line width=1.5pt, dashed, blue] (R1Start |- 0, 1.25) -- (R1End |- 0, 1.25);
    \draw[|-|, line width=1.5pt, dashed, blue] (R2Start |- 0, 1.25) -- (R2End |- 0, 1.25);

\end{tikzpicture}
\caption{An illustration for \cref{lem:finger} with two memory states $\phi,\phi'$ (corresponding to $S_{3,5}$ and $S_{3,4}$ in \cref{defn:sizeprof} respectively). Two maximal changed intervals are marked in dashed blue. The first interval contains objects of sizes $\{a,a,b,c\}$ in both $\phi, \phi'$. The second interval contains the finger objects of both $\phi,\phi'$.
} 
\label{fig:finger}
\end{figure}
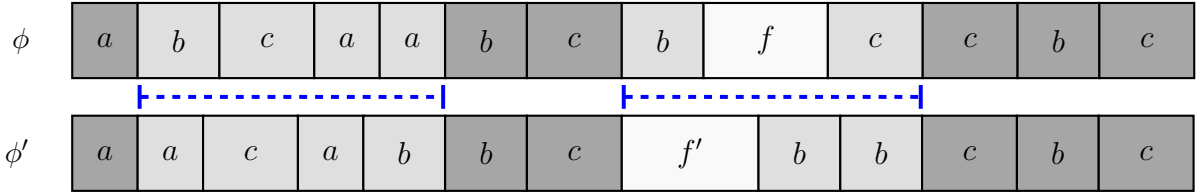

This property is reminiscent of our previous proof for the prefix setting (\cref{prop:warmup}): local permutation of objects may happen anywhere in memory, but insertions and deletions can only happen near the rightmost gap (or, in the non-prefix setting here, the finger object). 
However, unlike in the prefix setting, here the finger object may not always stay at the same place throughout the process. 
If we could somehow constrain the finger object to remain close to a location $p$, then we could adapt the previous potential-based proof by defining the potential of a size-$a$ object at location $\lambda$ to be $\approx |p-\lambda|$ (see the proof of  \cref{lem:a-sum}).

It remains to describe how to constrain the finger object in a relatively small interval.
We modify our earlier hard instance to include 
$Q$ size-$c$ objects, in addition to the $\Theta(M/b)$ size-$b$ objects at the beginning, where $Q$ is polynomially smaller than $M/b$.
As before, we gradually replace the size-$b$ objects by size-$a$ objects. However, we interrupt this process at a random time step unknown to the allocator, and start a ``surprise inspection'' instead, where we gradually delete all the size-$c$ objects.
 To stay prepared for the surprise inspection, the allocator must always keep all the size-$c$ objects close to the finger object; otherwise, it would be expensive to delete all the size-$c$ objects as demanded, since deletions can only happen near the finger object as we mentioned earlier (this strategy is formally implemented in the proof of  \cref{lem:c-clustered}).
Consequently, it is too expensive for the allocator to move the finger object significantly far away,  since it would have to carry along a large number of size-$c$ objects to move together with the finger object (this is formally proved in \cref{lem:c-position-sum} via a similar potential argument).
Therefore, the finger object is always constrained in a small interval as desired.

The proof strategy described above yields a deterministic worst-case $M^{\Omega(1)}$ lower bound on the overhead factor. Since our hard update sequence is oblivious,
the same proof also applies to randomized allocators with high-probability guarantees against an oblivious adversary. In fact, our proof gives a polynomial lower bound already for the $L^2$ norm of the overhead factor, as stated in \cref{thm:mainsecondmoment}. 

\subsection{Related works}
\label{sec:related}

\paragraph*{Memory (Re-)allocation.}
Bender, Farach-Colton, Fekete, Fineman, and Gilbert \cite{BenderFFFG17} defined and studied the Memory Reallocation problem in the \emph{cost-oblivious} setting, where the cost of inserting, deleting, and moving an object of size $\mu$ is an arbitrary unknown subadditive function~$c(\mu)$. 
This is substantially more general than many other works (including ours), e.g., \cite{NaorT01,Kuszmaul23,Farach-ColtonKS24}, which assume $c(\mu)=\mu$.

Naor and Teague \cite{NaorT01} and Kuszmaul \cite{Kuszmaul23} studied Memory Reallocation (for both unit-size and variable-size objects) in \emph{history-independent} settings.
In the strongly history-independent setting, the allocation of a set $S$ of labeled objects should only depend on the set $S$ and the randomness of the allocator.
Kuszmaul \cite{Kuszmaul23} showed that, even for the unit-size case, a strongly history-independent allocator must incur expected reallocation overhead $\Omega(\log\eps^{-1}/\log \log \eps^{-1})$, nearly matching the previous $O(\log \eps^{-1})$ upper bound by Berger, Kuszmaul, Polak, Tidor, and Wein \cite{BergerK0TW22}.

The very recent work of Bender, Conway, Farach-Colton, Koml\'os, Kuszmaul, and Wein \cite{nicole} studied a variant of memory allocation (without reallocation) in which each object being inserted is allowed to be \emph{fragmented}  into up to $k$ contiguous pieces. 
In this request fragmentation setting, they bypassed the logarithmic barrier for the memory competitive ratio in the classic $k=1$ case \cite{Robson71,Robson74,LubyNO96,nicole}.
They also determined the optimal memory competitive ratio in this setting.

\paragraph*{Sunflower lemma.}
The classic Erd\H{o}s--Rado sunflower lemma \cite{erdosrado} was quantitatively improved in a recent breakthrough by Alweiss, Lovett, Wu, and Zhang \cite{jiapeng}, and further refined by \cite{rao20,tao2020sunflower,
frankston2021thresholds, noteonsunflowers}; see the surveys by Rao \cite{raosurvey2023,raosurvey2025}.
The recent advances on sunflower lemmas led to the resolution of the Kahn--Kalai conjecture \cite{park2024proof}.

In theoretical computer science, the sunflower lemma has found applications in data structure lower bounds \cite{FrandsenMS97,GalM07,RamamoorthyR18}, circuit lower bounds \cite{razborov85,AlonB87,Rossman14,CavalarKR22,BlasiokM25,RezendeV25,cavalar2025monotone} and lifting theorems \cite{LovettMMPZ22}, parameterized algorithms \cite{Marx05}, sparsifying or analyzing DNF formulas \cite{GopalanMR13,LovettZ19,LovettSZ19,AkmalW21,Tantau22}, etc.
See \cite[Section 1.2 of the conference version]{AlweissL0Z20conference} for a more comprehensive list of references.

\paragraph{Subset sums and algorithm design.}
Our key lemma is extracted from Erd\H{o}s and S\'ark\"ozy's proof 
that the subset sums of any $S\subseteq [n]$ must contain an arithmetic progression of length $\Omega(|S|/\log^2 n)$ \cite{erdossarkozy}.
Schoen~\cite{schoenarithmeticprogression} subsequently improved this bound to $\Omega(|S|/\log n)$, but his proof does not yield the equal-sum disjoint subsets (crucial for our approach) as \cite{erdossarkozy}'s proof does.
The results of \cite{erdossarkozy,schoenarithmeticprogression} can be greatly improved in the regime where $|S|\ge n^{c}$ for constant $c>0$, e.g., \cite{alon1988sums,freiman93,sarkozy2,szemeredivu,apsumsets,fox}.
 Some of these results have found algorithmic applications in recent pseudopolynomial and approximation algorithms for 0-1 Knapsack and related problems, e.g., \cite{GalilM91,BringmannW21,ChenLMZsoda24,Bringmann24,Jin24,ChenLMZstoc24partition,clmzfocs24,CMZ25}. 
Additive combinatorics of subset sums has also led to faster algorithms for  Bin Packing \cite{NederlofPSW21,anticon}.

Very recently, motivated by questions related to Subset Sum algorithms, Chen, Mao, and Zhang \cite{soda26chen} designed a fast algorithm for constructing the equal-sum disjoint subsets as given by \cite{erdossarkozy}'s proof. Their result may potentially be helpful in designing a time-efficient implementation of our allocator; see further discussions in \cref{sec:con}.

\subsection*{Paper organization} 
\cref{sec:prelim} gives some useful preliminaries.
\Cref{sec:grouping} proves the key combinatorial lemmas used by our allocator. \Cref{sec:ds} describes our allocator. \Cref{sec:lb} proves the lower bounds. We conclude with several open questions in \cref{sec:con}.

\section{Preliminaries}
\label{sec:prelim}
\paragraph*{Notations.}
Denote $[n] = \{1,2,\dots,n\}$ and $\Z^+ = \{1,2,\dots\}$.
All logarithms are base two unless otherwise specified.

For $b>0$, let $a\bmod b$ denote the unique number $r\in [0,b)$ such that $a-r$ is an integer multiple of $b$.

We often write $A \sqcup B$ instead of $A \cup B$ to emphasize that the sets $A$ and $B$ are disjoint.

\paragraph*{On the definition of load factor.}
We have defined the load factor $1-\eps$ (where $\eps>0$) to mean that the total size of live objects at any time is at most $(1-\eps)M$, where $M$ is the number of memory slots.
Another definition in the literature (e.g., \cite{Kuszmaul23}) relaxes this upper bound to $\lfloor (1-\eps)M \rfloor + 1$;
in particular, when $M\le 1/\eps$, this upper bound becomes $M$, i.e., the memory is allowed to be full. 
The following \cref{obs:makememoryfull} shows that
these two definitions are actually equivalent up to changing $\eps$ by a constant factor. This fact will be convenient for proving lower bounds in \cref{sec:lb}.
\begin{observation}
   \label{obs:makememoryfull}
For $\eps>0$, if there is an allocator $\caA$ that (for all $M \in \Z^{+}$) handles total live objects size $\le (1-\eps)M$ with overhead $f(\eps)$, then, 
\begin{enumerate}
    \item  there is an allocator $\caA'$ that (for all $M\in \Z^{+}$) handles total live objects size $\le \lfloor (1-\eps)M\rfloor + 1$ with overhead $f(\eps/3)$, and
    \label{obs:item1}
    \item there is an allocator $\caA''$ that (for all $M\in \Z^{+}$) works even when all the $M$ memory slots can be full, with overhead $f(\frac{1}{3M})$.
    \label{obs:item2}
\end{enumerate}
\end{observation}
\begin{proof}
We first prove \cref{obs:item1}.
   Similar proof ideas already appeared in \cite{Kuszmaul23}.

  We describe how to design the desired allocator $\caA'$ given an implementation of $\caA$.
  Suppose $\caA'$ receives an input instance $I'$ with $M'$ memory slots $\{0,1,\dots,M'-1\}$ and maximum total live objects size $\le \lfloor (1-\eps)M'\rfloor +1$.

  Based on instance $I'$, define an input instance $I$ for $\caA$ as follows: there are $M\coloneqq 2M'+1$ memory slots $\{0,1,\dots,M-1\}$. 
  For every insertion/deletion of size $\mu'(x)\in \Z^{+}$ in the instance $I'$, we correspondingly create an insertion/deletion of size $\mu(x)\coloneqq 2\cdot \mu'(x)$ in $I$.

   We run $\caA$ on the constructed instance $I$, and translate the allocation of $\caA$ to our allocator $\caA'$ by the following rule: whenever $\caA$ decides to allocate object $x$ to the memory slots $[\phi(x), \phi(x) +\mu(x))$ (where integer $\phi(x) \in [0, M-\mu(x)] = [0,2M'+1-2\mu'(x)]$), we let $\caA'$ allocate $x$ to the memory slots $[\phi'(x), \phi'(x) +\mu'(x))$ where $\phi'(x)\coloneqq \lfloor \phi(x)/2\rfloor$. One can verify that:
   \begin{itemize}
      \item $ \phi'(x)\in [0,M'-\mu'(x)]$, i.e., the allocated memory slots for $x$ in $I'$ are in $[0,M')$, and
      \item $[\phi(x_1),\phi(x_1)+\mu(x_1)) \cap [\phi(x_2),\phi(x_2)+\mu(x_2)) =\varnothing$ implies $[\phi'(x_1),\phi'(x_1)+\mu'(x_1)) \cap [\phi'(x_2),\phi'(x_2)+\mu'(x_2)) =\varnothing$, i.e., the disjointness of the allocated intervals is preserved. 
   \end{itemize}
Thus, $\caA'$ is a valid allocator for the instance $I'$.
   
   The maximum total live object size of $I$ equals twice that of $I'$, and the memory size in $I$ is $M=2M'+1$. Therefore, the load factor of $I$ is
   \[\frac{2\cdot (\lfloor (1-\eps)M'\rfloor + 1)}{2M'+1}= 1 - \frac{\lceil \eps M'\rceil - 0.5}{M'+0.5} \le 1 - \frac{\eps M'/2}{M'+0.5} \le 1-\eps/3. \]
Hence, $\caA$ can achieve overhead $f(\eps/3)$ on instance $I$.  

Since the movement of our allocator $\caA'$ is caused by the movement of $\caA$ with the same moving objects (with size scaled by two, which does not affect the overhead factor), the overhead factor of $\caA'$ on instance $I'$ is also $f(\eps/3)$. This finishes the proof of \cref{obs:item1}.

We now derive \cref{obs:item2} as an immediate consequence of \cref{obs:item1}. 
Given an input instance with $M$ memory slots and maximum total live objects size $\le M$, we can directly solve it using the allocator in \cref{obs:item1} with $\eps$ set to $1/M$ so that $\lfloor (1-\eps)M\rfloor + 1 = M$. By \cref{obs:item1}, the overhead factor is $f(\eps/3) = f(\frac{1}{3M})$. 
\end{proof}

\paragraph*{Real-interval setting.}
When proving our upper bound in \cref{sec:ds}, to avoid excessive use of floors and ceilings, we use the following real-interval memory model introduced in \cite{Kuszmaul23} and \cite{Farach-ColtonKS24}: 
The memory is the \emph{real interval} $[0,M)$
instead of discrete memory slots $\{0,1,\dots,M-1\}$. Each object $x$ has a real-number size $\mu(x)$, so that the total size of live objects at any time is at most $(1-\eps)M$, where $1-\eps$ ($\eps>0$) is the load factor.
We are allowed to allocate $x$ at any real location $\phi(x)\in [0,M-\mu(x)]$, subject to the constraint that the occupied  intervals $[\phi(x),\phi(x)+\mu(x))$ are pairwise disjoint for all live objects $x$. 
The exact value of $M>0$ is not important in this setting, and can be rescaled to any positive real number.

\begin{observation}
   \label{obs:realinterval}
   If there is an allocator $\caA$ in the real-interval setting with overhead $f(\eps)$, then there is an allocator $\caA'$ that (for all $M\in \Z^{+}$) solves the original setting with $M$ discrete memory slots with overhead $f(\eps)$.
\end{observation}
\begin{proof}
Represent the memory by the real interval $[0,M)$. 
We now proceed in a similar way to the proof of \cref{obs:makememoryfull}.
Given an update sequence with load factor $1-\eps$ for the original discrete setting with $M$ memory slots, we feed the same sequence to the allocator $\caA$ which operates on the real interval $[0,M)$.
In particular, every object  in this update sequence has integer size.
Whenever $\caA$ allocates an object $x$ to the interval $[\phi(x),\phi(x)+\mu(x))$, in the discrete setting we allocate it to the memory slots indexed by $\lfloor \phi(x)\rfloor, \lfloor \phi(x)\rfloor + 1, \dots ,\lfloor \phi(x)\rfloor + \mu(x)-1$. Using $\mu(x)\in \Z$, one can verify that this transformation preserves the disjointness of the allocated intervals, and all the used memory slots are from $\{0,1,\dots,M-1\}$.

The load factor of this update sequence in the real-interval setting is also $\frac{(1-\eps)M}{M}=1-\eps$, so $\caA$ achieves overhead factor $f(\eps)$. Thus, our allocator in the discrete setting also achieves overhead factor $f(\eps)$.
\end{proof}

\section{Combinatorial lemmas}
\label{sec:grouping}

A \emph{sunflower with $p$ petals} is a family of $p$ sets whose pairwise intersections are identical (called the \emph{core}). 
The following is the state-of-the-art bound on the sunflower lemma, proved by Bell, Chueluecha, and Warnke \cite{noteonsunflowers} (based on the recent breakthrough of \cite{jiapeng} and subsequent refinements \cite{rao20,tao2020sunflower,
frankston2021thresholds}).

\begin{lemma}[Sunflower lemma \cite{noteonsunflowers}]
  There is a constant $C\ge 4$ such that the following holds for all integers $p,k\ge 2$. Any family of at least $(Cp\log k)^k$ distinct $k$-element sets must contain a sunflower with $p$ petals.
\label{lem:sunflower}
\end{lemma}

Following Erd\H{o}s and S\'{a}rk\"{o}zy \cite{erdossarkozy}, we use the sunflower lemma to prove the following
\cref{lem:grouping-temp}.  (If one uses the original Erd\H{o}s--Rado bound of $1 + k!(p-1)^k$ \cite{erdosrado} instead of the improved bound in 
\cref{lem:sunflower}, the last bound in \cref{lem:grouping-temp} would worsen to $p \ge \frac{n}{C\log^2 w}$.)
\begin{lemma}
\label{lem:grouping-temp}
There is a constant $C$ such that for any $w\ge 4$ and any sequence of $n$ positive integers $a_1,a_2,\dots,a_n \in [w]$, there exist disjoint subsets $B_1,B_2,\dots,B_p\subseteq [n]$ such that 
\begin{itemize}
    \item The sums $\sum_{i\in B_j}a_i$ are equal for all $j\in [p]$, 
    \item $|B_1|=|B_2|=\dots =  |B_p|\le C\log  w$,  and
    \item $p\ge  \frac{n}{C\log w\log \log w}$.
\end{itemize}
\end{lemma}

\begin{proof}
The proof is due to Erd\H{o}s and S\'{a}rk\"{o}zy \cite{erdossarkozy} (see also recent exposition in \cite{raosurvey2023}).
   We may assume $w$ and $n$ are large enough (otherwise, the lemma immediately holds for a large enough constant $C$). Let $k =  \lceil 2\log w\rceil $. 
   For $s \le nw$, define set family \[\caS_s \coloneqq \left \{B\in \binom{[n]}{k}: \sum_{i\in B} a_i = s\right \}.\]
  Since $0< \sum_{i\in B}a_i \le kw$ holds for all $B\in \binom{[n]}{k}$, we clearly have $\binom{n}{k} = \sum_{1\le s\le kw}|\caS_s|$. By averaging, there exists $s$ such that 
  \[|\caS_s| \ge \frac{\binom{n}{k}}{kw} \ge \frac{(n/k)^k}{kw} \ge \left (\frac{n}{2k}\right )^k, \]
  where the last inequality follows from $2^k\ge kw$ by our choice of $k = \lceil 2\log w\rceil $. 
  
  Choose $p = \lfloor \frac{n}{2Ck\log k}\rfloor \ge \frac{n}{C'\log w \log\log w}$, where $C'>0$ is some constant.
   Then, the above inequality implies $|\caS_s|\ge (Cp\log k)^k$. Hence, by 
  \cref{lem:sunflower}, $\caS_s$ contains a sunflower with $p$ petals $B'_1,B'_2,\dots,B'_p$. 
Let the core be $R = B'_1 \cap \dots \cap B'_p$, and define $B_1\coloneqq B'_1\setminus R, \dots, B_p \coloneqq B'_p \setminus R$. Clearly, for all $j\in [p]$,  $\sum_{i\in B_j}a_i = \sum_{i\in B'_j}a_i - \sum_{i\in R}a_i = s - \sum_{i\in R}a_i$ and  $|B_j| = k - |R|$, so the requirements are indeed all satisfied.
\end{proof}

The following lemma is a simple refinement of the previous one. It allows us to shave off a logarithmic factor in our later application. 
\begin{lemma}
   \label{lem:grouping-main}
In the same setup as  \cref{lem:grouping-temp}, one can additionally achieve $|B_1|+|B_2|+\dots +|B_p| \ge \frac{n}{C\log \log w}$.
\end{lemma}
\begin{proof}
   We again assume $n,w$ are large enough.
   We iteratively apply \cref{lem:grouping-temp} as follows. Define $p^* = \lfloor \frac{n}{2C\log w\log \log w}\rfloor$, where $C$ is the constant from \cref{lem:grouping-temp}.

  First, apply \cref{lem:grouping-temp} to $a_1,\dots,a_n$, and obtain disjoint subsets $B_1^{(1)},\dots,B^{(1)}_{p_1} \subseteq [n]$ such that $|B_j^{(1)}| = k_1$ and $\sum_{i\in B_j^{(1)}}a_i = s_1$ for all $j\in [p_1]$. 
  Since $p_1\ge p^*$, we can decrease $p_1$ down to $p^*$.
  
  Then, apply \cref{lem:grouping-temp} to the remaining numbers in $a_1,\dots,a_n$ indexed by $[n] \setminus (B^{(1)}_1\sqcup \dots \sqcup B^{(1)}_{p_1})$, and obtain disjoint subsets $B_1^{(2)},\dots,B^{(2)}_{p_2} \subseteq [n]\setminus (B^{(1)}_1\sqcup \dots \sqcup B^{(1)}_{p_1})$ such that $|B_j^{(2)}| = k_2$ and $\sum_{i\in B_j^{(2)}}a_i = s_2$ for all $j\in [p_2]$.
  As long as there remain at least $n/2$ numbers, we still have $p_2\ge p^*$, so we can decrease $p_2$ to $p^*$. 

Repeat this procedure until the $m$-th iteration, where $m$ is the smallest integer such that $k_1+k_2+\dots +k_m \ge C\log w$. We have $k_1+k_2+\dots +k_m < C\log w + k_m \le 2C\log w$. 
Note that the total size of subsets removed so far is \[\sum_{i=1}^m\sum_{j=1}^{p^*} |B_j^{(i)}| = p^*(k_1+k_2+\dots+k_m) \le \frac{n}{2C\log w\log \log w} (2C\log w)\le \frac{n}{\log \log w}\le \frac{n}{2},\]
assuming $w$ is large enough. 
 Hence, when invoking \cref{lem:grouping-temp}, we always have at least $n/2$ objects remaining, so $p_1,p_2,\dots$ are indeed lower-bounded by $p^*$.

  Finally, we return $B_j\coloneqq \bigcup_{i=1}^m B_j^{(i)}$ for all $j\in [p^*]$. Clearly, we still have that all $\sum_{i\in B_j}a_i$ are equal, and that $|B_1|=\dots = |B_{p^*}|= k_1+k_2+\dots +k_m \in [C\log w, 2C\log w]$. Thus, $|B_1|+\dots +|B_{p^*}| \ge p^*\cdot C\log w \ge \frac{n}{C'\log \log w}$ for some constant $C'>0$. Therefore, all requirements are satisfied (with a possibly different constant $C$).
\end{proof}

Finally, we iteratively apply \cref{lem:grouping-main} to partition \emph{all} the given integers into logarithmic-sized subsets which have logarithmically many distinct sums.
This is the lemma that we will use for designing our allocator.

\begin{lemma}
\label{lem:classifier}
There exist constants $C_1,C_2$ such that the following holds. For any $w\ge 4$, and any sequence of $n$ positive integers $a_1,a_2,\dots,a_n \in [w]$, there exist 
a partition 
\begin{equation}
   \label{eqn:types}
[n] = T_1 \sqcup T_2\sqcup \dots \sqcup T_{\ell},  \ \ell\le C_1\log 2n\log\log w, 
\end{equation}
 and $\ell$ positive integers $s_1,\dots,s_\ell$, where each $T_j$ is further partitioned into $p_j$ subsets,
\begin{equation}
   \label{eqn:bundles-of-a-type}
T_j = B_{j,1}\sqcup B_{j,2} \sqcup \dots \sqcup B_{j,p_j},
\end{equation}
such that for all $j$:
\begin{enumerate}
       \item For all $1\le k\le p_j$, $\sum_{i\in B_{j,k}}a_i = s_j$.
          \label{item:samesize}
    \item For all $1\le k\le p_j$, $|B_{j,k}| \le C_2\log w$.
          \label{item:smallbundle}
       \item  
          For all $1\le m\le n$, the number of $j \in [\ell]$ such that $p_j \ge m$ is at most $C_1\log(2n/m)\log \log w$.
          \label{item:pjtail}
\end{enumerate}
\end{lemma}
We remark that the first two properties in \cref{lem:classifier} are more important. 
\cref{item:pjtail} only serves to shave a near-logarithmic factor in our later application.

\begin{proof}
  Starting from all the indices $[n]$, we iteratively apply \cref{lem:grouping-main} to build the partition \cref{eqn:types}.
   Each application of \cref{lem:grouping-main} returns one set $T_j$ together with its decomposition \cref{eqn:bundles-of-a-type} which satisfies \cref{item:samesize} and \cref{item:smallbundle} by definition. 
   After each application, we remove the already processed part $T_j$, and continue with the remaining indices, until all of $[n]$ have been partitioned.
   
   By the guarantee of \cref{lem:grouping-main}, each iteration removes a subset $T_j$ which contains an $\Omega(1/(\log \log w))$ fraction of the remaining indices. Consequently, after $O(\log (2n/m)\log \log w)$ iterations, the number of remaining indices drops below $m$, and after that we must have $p_j\le |T_j|< m$. This proves \cref{item:pjtail}. In particular, the total number of iterations is $\ell = O(\log n\log \log w)$.
\end{proof}

\section{The allocator}
\label{sec:ds}

In this section, we present our allocator with polylogarithmic overhead, proving \cref{thm:main}. 
By \cref{obs:realinterval},  we model the memory as the real interval $[0,M)$ throughout this section. 

Our main lemma is an allocator that achieves polylogarithmic overhead when the objects are not tiny.
\begin{lemma}
   \label{lem:largemain}
   Let $\ell \ge 2$ be an integer.
   In the real-interval setting with load factor $1-\eps$ where all objects $x$ have sizes $\mu(x)>2^{-\ell}M$,
   there is a resizable allocator with (worst-case) expected overhead
$O(\ell^3(\log\eps^{-1}) (\log \ell + \log \log \eps^{-1})^2)$.
\end{lemma}

Combined with previously known results, \cref{lem:largemain} implies our main theorem.
\begin{proof}[Proof of \cref{thm:main} assuming \cref{lem:largemain}]
Kuszmaul \cite{Kuszmaul23} gave a  resizable allocator for objects of size $\le \eps^{4}M$  with worst-case expected overhead $O(\log \eps^{-1})$.
Farach-Colton, Kuszmaul, Sheffield, and Westover
\cite[Section 4.2]{Farach-ColtonKS24} showed that, given another resizable allocator that handles objects of size larger than $\eps^{4}M$ with worst-case expected overhead $f(\eps)$, one can combine it
with Kuszmaul's allocator to obtain a resizable allocator that handles arbitrary object sizes and achieves worst-case expected overhead $O(f(\eps) + \log(\eps^{-1}))$.
Since \cref{lem:largemain} achieves $f(\eps)\le O(\log^4\eps^{-1}\cdot (\log \log \eps^{-1})^2)$ for $\ell = O(\log \eps^{-1})$, together they imply \cref{thm:main}.
\end{proof}

The rest of this section is devoted to the proof of \cref{lem:largemain}. 
As introduced in the overview section, the proof is a combination of our bundling idea based on \cite{erdossarkozy} (see \cref{sec:grouping}) and the substitution strategy with periodic rebuilds from \cite{Farach-ColtonKS24} (as illustrated in the warm-up proof of \cref{prop:fewtype}).
Additionally, we need to relax the bounded-ratio assumption of \cref{prop:fewtype} in the same way as \cite{Farach-ColtonKS24}.

For simplicity, we assume the load factor is $1-2\eps$ instead of $1-\eps$.
We also assume $\eps^{-1}$ is an integer. Both assumptions can be justified by scaling $\eps$ by at most a constant factor.

We rescale the length of the real memory interval $[0,M)$ to $M \coloneqq 2^\ell$.
Then, the assumption in \cref{lem:largemain} implies that each object $x$ has size $\mu(x)>1$.

\subsection{Basic definitions}

\paragraph*{Scales and inflation.}
A \emph{scale} refers to an interval $(2^{i},2^{i+1}]$ of object sizes, where $i\in \Z$. 
For the purpose of \cref{lem:largemain}, the relevant scales are $(1,2],(2,4],\dots, (2^{\ell-1},2^{\ell}]$, i.e., $i\in \{0,1,\dots,\ell-1\}$.

For an object $x$ with size $\mu(x) = s\in (2^{i},2^{i+1}]$, we inflate its size to $\mu(x)\coloneqq \lceil \frac{s}{\eps \cdot 2^{i+1}} \rceil \cdot  \eps \cdot 2^{i+1}$.
Since $\eps^{-1}$ is an integer, the inflated size is $\eps\cdot 2^{i+1}$ multiplied by an integer in $(\eps^{-1}/2,\eps^{-1}]$, and still lies in the interval $(2^{i},2^{i+1}]$.
 The inflation increases $\mu(x)$ by at most a factor of $(1+2\eps)$, so the new load factor becomes $(1-2\eps) \cdot (1+2\eps) < 1$, i.e., the total size of live objects remains smaller than $M$.
The inflation affects the overhead factor by at most $1+2\eps = O(1)$ multiplicatively.
We always use $\sz(x)$ to denote the size of $x$ after inflation.
Clearly, an allocation of the inflated objects automatically gives an allocation of the original objects.

After inflating the objects, our allocator maintains the prefix property, namely that all objects are stored without any gaps in between, starting from the left boundary of the memory.  
This property immediately implies that our allocator is resizable.

\paragraph*{Bundles and types.}
Our allocator partitions the live objects into \emph{bundles} of various \emph{types}.
The type of a bundle $B$ is denoted by $\typ(B)$.
Objects in the same bundle are always allocated contiguously in memory.
Objects in the same bundle must have the same scale $(2^i,2^{i+1}]$. 
Any two bundles $B,B'$ of the same type must have equal size, $\sum_{y\in B}\mu(y) = \sum_{y\in B'}\mu(y)$.
From time to time, the allocator may unbundle old bundles and form new bundles.

Given a subset $Y$ of objects,
we use the procedure $\crea(Y)$ (\cref{alg:createbundles}) to partition $Y$ into bundles, and assign  newly created types to these bundles. 
This procedure invokes \cref{lem:classifier} for each scale $(2^i,2^{i+1}]$ separately  to bundle the objects of that scale in $Y$.
In \cref{lem:classifier},
each subset $B_{j,k}\subset T_j$ in \cref{eqn:bundles-of-a-type} represents a bundle of objects, and $T_j$ corresponds to the type of that bundle.

\begin{algorithm}
\DontPrintSemicolon
\caption{$\crea(Y)$, where $Y$ is a set of objects}
\label{alg:createbundles}
$\caB \gets \varnothing $ \tcp{the collection of created bundles}
\For{$i\in \{0,1,\dots,\ell-1\}$}{ %
   $Y_i \coloneqq \{y\in Y: \mu(y) \in (2^{i},2^{i+1}]\}$\\
   Let $Y_i = \{y_1,y_2,\dots,y_{|Y_i|}\}$, and define $a_m\coloneqq \mu(y_m)/(\eps\cdot 2^{i+1})$.\\
   Apply \cref{lem:classifier} with $w\coloneqq \eps^{-1}$ to integers $a_1,\dots,a_{|Y_i|}\in [w]$, and obtain the partition of $[|Y_i|]$ as in \cref{eqn:types,eqn:bundles-of-a-type}
\\
\For{each $T_j$ in \cref{eqn:types}}{
   $\tau_{\mathrm{new}} \gets \tau_{\mathrm{new}} + 1$ \tcp{Create a new type indexed by $\tau_{\mathrm{new}}$ (a global counter).}
   \For{each $B_{j,k}\subset T_j$ in \cref{eqn:bundles-of-a-type}}{
      Create a bundle $B \coloneqq \{ y_m : m \in B_{j,k}\}$ with $\typ(B) \coloneqq \tau_{\mathrm{new}}$\\
      $\caB \gets \caB \sqcup \{B\}$.
   }
}
}
\Return{$\caB$}
\end{algorithm}

\paragraph*{Levels.}

The live objects allocated in memory are partitioned into $\ell$ \emph{levels}, where each level consists of a contiguously allocated (possibly empty) set of bundles.
The levels are indexed decreasingly by $\ell-1,\ell-2,\dots,1,0$ from left to right.

We maintain that level $j$ can only contain objects $x$ with $\mu(x)\le 2^{j+1}$. In other words, an object of scale $(2^{i},2^{i+1}]$ can only be placed in levels $\ell-1,\ell-2,\dots,i+1,i$. 
In particular, an object of the smallest scale $(1,2]$ may be placed in any level.

For a bundle type $\tau$, we define its \emph{leftmost level}, denoted $\lowlevel(\tau)$, as the largest $j\le \ell-1$ such that level $j$ contains a bundle of type $\tau$.          (If no type-$\tau$ bundles currently exist in memory, then $\lowlevel(\tau) \coloneqq  -\infty$.)

\subsection{Implementation}
At the very beginning, run $\textsc{Initialize}$ (\cref{algo:initialize}).
\paragraph*{Insertion and deletion.}
Now we describe how our allocator handles updates (i.e., insertions and deletions).
The implementations are given in \cref{alg:insert} and \cref{alg:delete}, respectively.
To slightly unify the descriptions of these two procedures, we introduce a global variable $Y$ to represent the set of objects that are yet to be added to memory. More specifically: 
\begin{itemize}
    \item  For insertion, we set $Y$ to contain the current object being inserted. 
    \item  In the case of deleting $x$, whose scale is $(2^{i^*},2^{i^*+1}]$, we first perform the substitution strategy, so that the bundle $B$ containing $x$ now appears in level $i^*$ (namely the rightmost level allowed to contain an object of this scale).
    Then, remove the entire bundle $B$ from memory (creating a temporary gap in level $i^*$), and set $Y \gets B\setminus \{x\}$, indicating that these objects should be moved back into memory later.
\end{itemize}

Both types of updates finish by invoking the $\rebu(j^*)$ procedure (\cref{alg:rebuild}) for some $j^* \in \{0,1,\dots,\ell-1\}$, whose behavior is to reorganize objects in levels $j^*,j^*-1,\dots,1,0$ together with the objects from $Y$, without modifying the levels to the left of $j^*$.
Here, $j^*$ is determined in a way so that levels further to the left are less frequently rebuilt, as follows:
For each scale $(2^{i}, 2^{i+1}]$, maintain a counter $c_i$, initialized uniformly at random from $\Z \cap [0,2^{\ell-1-i})$.\footnote{This is the only randomized part in our proof of \cref{lem:largemain}. If we instead initialize $c_i\gets 0$ for all $i$, then we would obtain a \emph{deterministic} allocator for object sizes $[\poly(\eps)M,M)$
 with polylogarithmic \emph{amortized} overhead. We do not know how to extend this amortized derandomization to all object sizes in $(0,M)$, since the tiny objects are handled by Kuszmaul's allocator \cite{Kuszmaul23} which seems inherently randomized.}
At the end of each update of object scale $(2^{i^*}, 2^{i^*+1}]$, we first increment the counter $c_{i^*}$ and then invoke $\rebu(j^*)$ for the largest integer $j^*\in [i^*,\ell-1]$ such that $2^{j^*-i^*}$ divides $c_{i^*}$.

\begin{algorithm}
\DontPrintSemicolon
\caption{$\textsc{Initialize}$}
\label{alg:initialize}
\For{$i\in \{0,1,\dots,\ell-1\}$}{$c_i \gets $ random integer from $[0,2^{\ell-1-i}) $ \tcp{  a global counter}}
$\tau_{\mathrm{new}} \gets 0 $ \tcp{  a global counter}
\label{algo:initialize}
\end{algorithm}

\begin{algorithm}
\DontPrintSemicolon
\caption{$\inse(x)$}
\label{alg:insert}
Suppose $\sz(x) \in (2^{i^*},2^{i^*+1}]$\\
$c_{i^*}\gets c_{i^*}+1$\\
$Y \gets \{x\}$ \tcp{ a global variable }
$\rebu(j^*)$, where $j^*$ is the largest integer $j^*\in [i^*,\ell-1]$ such that $2^{j^*-i^*}$ divides $c_{i^*}$
\end{algorithm} 

\begin{algorithm}
\DontPrintSemicolon
\caption{$\dele(x)$}
\label{alg:delete}
Suppose $\sz(x) \in (2^{i^*},2^{i^*+1}]$\\
$c_{i^*}\gets c_{i^*}+1$\\
Let $B$ be the bundle containing $x$ \\
\If{$B$ is not in level $i^*$}
{
\tcp{$B$ must be in levels $\{\ell-1,\ell-2,\dots,i^*+1\}$}
   Let $B'$ be a bundle in level $i^*$ such that $\typ(B')=\typ(B)$ (which must exist by \cref{lem:alwayshasqr})
   \label{line:find-bundle-to-swap}
   \\
   Swap $B$ and $B'$ in memory %
}
\tcp{Now, $B$ is in level $i^*$, and $x$ is in $B$.}
Remove $B$ from memory\\
$Y \gets B\setminus \{x\}$ 
\tcp{ a global variable }
$\rebu(j^*)$, where $j^*$ is the largest integer $j^*\in [i^*,\ell-1]$ such that $2^{j^*-i^*}$ divides $c_{i^*}$
\end{algorithm} 

\begin{algorithm}
\DontPrintSemicolon
\caption{$\rebu(j^*)$ where $0\le j^*\le \ell-1$}
\label{alg:rebuild}

$\caT_{\text{old}}  \coloneqq \{\text{type }\tau: \lowlevel(\tau)> j^*\}$\\
$\caB \gets \varnothing$ \tcp{a collection of bundles, initially empty}

\For{each bundle $B$ in levels $\{j^*,j^*-1,\dots,0\}$}{
   Remove $B$ from the memory \\
\lIf(\texttt{ // Unbundle $B$ and add all its objects to $Y$}){$\typ(B)\notin \caT_{\text{old}}$}{$Y \gets Y \sqcup B$} \lElse(\texttt{ // $B$ remains as an old bundle, which we add to $\caB$}){$\caB \gets \caB \sqcup \{B\}$}
}
\tcp{Now, the memory in levels $\{j^*,j^*-1,\dots,0\}$ has become empty.}
$\caB \gets \caB \sqcup \crea(Y)$, and then reset $Y\gets \varnothing$.\\

\For{each bundle type $\tau$}{
$\caB_{\tau} \coloneqq \{B \in \caB: \typ(B)=\tau\}$\\
  Let $(2^{i},2^{i+1}]$ be the scale of objects in type-$\tau$ bundles \\
Define sequence $(n_i,n_{i+1},\dots,n_{j^*})$ by $n_{j} =\begin{cases} 
  2^{\max\{j-i,1\}} &  j\neq j^*\\
   +\infty & j=j^*
\end{cases}$ \tcp{$(n_i,n_{i+1},\dots,n_{j^*})= (2,2,4,8,16,\dots,+\infty)$  }
\For{$j\gets i,i+1,\dots,j^*$\label{line:putbundlesfor}}{If there are $n$ unassigned bundles in $\caB_{\tau}$, assign $\min\{n,n_{j}\}$ of them to level $j$ in memory\label{line:putbundles}}
}
\end{algorithm}

\paragraph*{Rebuild.}
During $\rebu(j^*)$ (\cref{alg:rebuild}), we rebuild the levels $j^*,j^*-1,\dots,1,0$ and also add the objects from $Y$ into them, without modifying the remaining levels $\ell-1,\ell-2,\dots,j^*+1$ on the left.
For a bundle $B$ in levels $\{j^*,j^*-1,\dots,1,0\}$, if there is another bundle $B'$ in levels $\{\ell-1,\ell-2,\dots,j^*+1\}$ with $\typ(B')=\typ(B)$, then we keep $B$ as an old bundle (since it may still be useful for performing the substitution strategy in the future); otherwise, we unbundle $B$ and add its objects to $Y$. 
We partition the objects in $Y$ into new bundles.
We then place the bundles (both new and old)  back into memory following a right-to-left greedy rule similar to that in the warm-up proof of \cref{prop:fewtype}. Here, we assign each bundle to one of the levels in
$\{j^*,j^*-1,\dots,1,0\}$; bundles within the same level may be placed in arbitrary order.

We have the following basic observations:
\begin{observation}
   \label{observation:rebuildsuffix}
 $\rebu(j^*)$ does not modify levels $\{\ell-1,\ell-2,\dots,j^*+1\}$ in memory. 
\end{observation}

\begin{observation}
After every update ($\inse(\cdot)$ or $\dele(\cdot)$) finishes, the memory satisfies the prefix property.
\end{observation}
\begin{proof}
Right before invoking $\rebu(j^*)$ in $\dele(x)$,  the levels $\ell-1,\ell-2,\dots,i^*+1$ satisfy the prefix property, because the deletion may only creates a gap in level $i^*$. Since $j^*\ge i^*$, after $\rebu(j^*)$, the memory satisfies the prefix property again.

For $\inse(x)$, the proof is straightforward.
\end{proof}

\subsection{Invariants}
Let $C_1\ge 1,C_2 \ge 1$ be the constants from \cref{lem:classifier}. 
Let $C_3$ be a constant (depending on $C_1,C_2$) such that 
\begin{equation}
 C_3 \ge 3C_1\log(12C_2C_3).
\end{equation}

Below, we state the three main invariants which the allocator satisfies at the end of every update ($\inse(\cdot)$ or $\dele(\cdot)$).
\cref{item:inva-fewtypes} ensures that there are few distinct types of bundles at any time.
The remaining two invariants are analogous to those appearing in the warm-up proof of \cref{prop:fewtype}:
\cref{item:inva-nottoobig} gives an upper bound on the total size of each level, and
\cref{item:inva-supply} guarantees enough supply of bundles of each type
for performing the substitution strategy.
\begin{property}
         \label{item:inva-fewtypes}
For every $i,j$ with $\ell \ge j\ge i\ge 0$, the number of distinct types among bundles of scale $(2^{i},2^{i+1}]$ in levels $\{\ell-1,\ell-2,\dots,j\}$ is at most 
\[C_3(\ell-j) (\log \ell + \log\log \eps^{-1})^2.\]
\end{property}

\begin{property}
         \label{item:inva-nottoobig}
         Let $s(i,j)$ denote the total size of objects of scale $(2^{i},2^{i+1}]$ that are in levels $\{j-1,j-2,\dots,i\}$. %
   
   Then, for every $i,j$ with $\ell\ge j> i\ge 0$, 
   \begin{equation}
      \label{eqn:ineqsij}
    s(i,j) \le  2^{j+1}\cdot C_2\log \eps^{-1} \cdot C_3\ell (\log\ell + \log\log \eps^{-1})^2  + (c_i\bmod 2^{j-i})\cdot 2^{i+1}. 
   \end{equation}
   In particular,
   \[s(i,j) \le 2^{j}\cdot 3C_2C_3\ell (\log\eps^{-1})(\log \ell + \log\log \eps^{-1})^2 - 2^{i+1}.\]
\end{property}

\begin{property}
         \label{item:inva-supply}
    Let $v(\tau,j)$ denote the number of type-$\tau$ bundles in levels $\{j-1,j-2,\dots,0\}$.
    
    Then, for every bundle type $\tau$ with object scale $(2^{i},2^{i+1}]$, and every $j$ such that $ \lowlevel(\tau) \ge j > i$,
   \begin{equation}
      \label{eqn:ineqvtj}
    v(\tau,j) \ge 2^{j-i} - (c_i\bmod 2^{j-i}) \ge 1.
   \end{equation}
\end{property}

\begin{corollary}
   \label{lem:alwayshasqr}
   At \cref{line:find-bundle-to-swap} of $\dele(x)$ (\cref{alg:delete}), the substitute bundle $B'$ always exists.
\end{corollary}
\begin{proof}
   Let $\tau = \typ(B)$, and recall that the scale of $\tau$ is $(2^{i^*},2^{i^*+1}]$.
   Since $B$ 
   can only be in levels $\{\ell-1,\ell-2,\dots,i^*+1\}$, we have $\lowlevel(\tau) \ge i^*+1$. Hence, by \cref{item:inva-supply} with $\tau$ and  $j=i^*+1$, we have $ v(\tau,i^*+1) \ge 1$.
   Note that $v(\tau,i^*+1)$ counts exactly the number of type-$\tau$ bundles in level $i^*$, since objects of scale $(2^{i^*},2^{i^*+1}]$ can only appear in levels $\{\ell-1,\ell-2,\dots,i^*\}$. Therefore, there exists at least one type-$\tau$ bundle $B'$ in level $i^*$. 
\end{proof}

We use time step $t$ to refer to the state at the end of the $t$-th update ($\inse(\cdot)$ or $\dele(\cdot)$) of the update sequence. 
Initially at time $0$, the memory is empty and \cref{item:inva-fewtypes,item:inva-nottoobig,item:inva-supply} all hold vacuously.
We now use induction on $t$ to prove that these invariants hold at every time step $t \ge 1$. 
\begin{proof}[Proof of \cref{item:inva-fewtypes}]
  To prove \cref{item:inva-fewtypes} for $i,j$ at any time step $t$, suppose the most recent execution of $\rebu(j^*)$ such that $j^* \ge j$ finished at time $t'\le t$; if no such execution of $\rebu(j^*)$ has occurred, let $t'=0$.
Since the levels $\{\ell-1,\ell-2,\dots,j\}$ have never been modified after time $t'$, it suffices to prove the claimed upper bound at time $t'$.
If $t'=0$, then the bound immediately holds. Henceforth, assume $t'\ge 1$.
  
Focus on the execution of $\rebu(j^*)$ at time $t'$.  We now separately bound the number of old types $\tau \in \caT_{\mathrm{old}}$ and new types $\tau \notin \caT_{\mathrm{old}}$ of scale $(2^{i},2^{i+1}]$, and add them up.
\begin{itemize}
   \item 
      Recall that $\tau$ is an old type  if and only if $\lowlevel(\tau)\ge j^*+1$. 
      Since \cref{item:inva-fewtypes} held at time step $t'-1$ for the levels $\{\ell-1,\ell-2,\dots,j^*+1\}$, this implies that the number of old types of scale $(2^{i},2^{i+1}]$ is at most $C_3(\ell-j^*-1)(\log \ell + \log\log \eps^{-1})^2$.

  \item It remains to bound the number of new types of scale $(2^{i},2^{i+1}]$ created by $\crea(Y)$.
We focus on the subset $Y_i\subseteq Y$ containing objects of scale $(2^{i},2^{i+1}]$ only.
By definition, all objects of $Y_i$ were in levels $\{j^*,j^*-1,\dots,i\}$ at time $t'-1$, with the possible exception of one additional object if the $t'$-th update is an insertion. 
Since \cref{item:inva-nottoobig} held for scale $(2^i,2^{i+1}]$ and levels $\{j^*,j^*-1,\dots,i\}$ at time $t'-1$, this implies that the objects in $Y_i$ have total size at most
  \[2^{j^*+1}\cdot 3C_2C_3\ell (\log\eps^{-1})(\log \ell + \log\log \eps^{-1})^2.\]
  Consequently,
  \[|Y_i|\le 2^{j^*-i}\cdot 6C_2C_3\ell (\log\eps^{-1})( \log \ell + \log\log \eps^{-1})^2.\]
  For a new type $\tau$,  if a level $j\in [i,j^*]$ receives at least one type-$\tau$ bundle at the end of the \textbf{for} loop in \cref{line:putbundlesfor,line:putbundles}, then the number of type-$\tau$ bundles, $|\caB_{\tau}|$, must be at least $1 + n_{j-1}+n_{j-2}+\dots + n_{i} = \begin{cases}
  1 + 2^{j-i} & i< j \\ 1 & i=j
  \end{cases} \ge 2^{j-i}$.
  On the other hand, by definition of $\crea(Y)$ and \cref{lem:classifier}, \cref{item:pjtail}, the number of types $\tau$ with at least $2^{j-i}$ bundles is at most 
   \begin{align*}
 & C_1\log\!\left(2|Y_i|/2^{j-i}\right)\log \log \eps^{-1}\\
 & \le C_1\log\!\left(2^{j^*-j}\cdot 12C_2C_3\ell (\log\eps^{-1})(\log \ell + \log\log \eps^{-1})^2\right)\log \log \eps^{-1}\\
 & \le C_1(j^*-j + \log(12C_2C_3) + 3\log\ell +3\log \log \eps^{-1})\log\log\eps^{-1}\\
 & \le C_3(j^*-j+1)(\log \ell + \log \log \eps^{-1})^2. 
  && \text{(by $C_3 \ge 3C_1\log(12C_2C_3)$)}
  \end{align*}
Thus, the number of new types of scale $(2^{i},2^{i+1}]$ in level $j$ is at most $C_3(j^*-j+1)(\log \ell + \log \log \eps^{-1})^2$.
\end{itemize}
Summing up the counts of old and new types proves the claimed upper bound $C_3(\ell-j) (\log \ell + \log\log \eps^{-1})^2$.
\end{proof}

\begin{proof}[Proof of \cref{item:inva-nottoobig}]
To prove the claimed bound on $s(i,j)$ at any time step $t$, suppose the most recent execution of $\rebu(j^*)$ such that $j^*\ge j$ finished at time $t'\le t$; if no such execution of $\rebu(j^*)$ has occurred, let $t'=0$.
By definition, $t'$ is no earlier than the most recent time when the counter $c_i$ became divisible by $2^{j-i}$. Hence, after time $t'$, there have been at most $(c_i\bmod 2^{j-i})$ updates of object scale $(2^i,2^{i+1}]$, where $c_i$ denotes the counter at current time $t$. 

Let $s'(i,j)$ denote the total size  of objects of scale $(2^i,2^{i+1}]$ in levels $\{j-1,j-2,\dots,i\}$ at time $t'$. If $t'=0$, then $s'(i,j)=0$. Now we bound $s'(i,j)$ in the $t'\ge 1$ case.
Focus on the execution of $\rebu(j^*)$ at time $t'$. By the placement rule at the end of $\rebu(j^*)$, since $j^*>j-1$, each bundle type $\tau$ of object scale $(2^i,2^{i+1}]$ contributes at most
$n_i+n_{i+1}+\dots + n_{j-1} =  2^{j-i}$
bundles to the levels $\{j-1,j-2,\dots,i\}$. 
By \cref{item:inva-fewtypes}, the number of distinct bundle types of scale $(2^{i},2^{i+1}]$ is at most $C_3\ell(\log \ell + \log \log \eps^{-1})^2$.
By \cref{lem:classifier}, every bundle created by $\crea(\cdot)$ contains at most $C_2\log \eps^{-1}$ objects. 
By multiplying these quantities together, we get the following upper bound:
\[  s'(i,j) \le C_3 \ell (\log \ell + \log \log \eps^{-1})^2\cdot 2^{j-i}   \cdot C_2 \log \eps^{-1} \cdot 2^{i+1}. \]

After time $t'$, the levels $\{\ell-1,\ell-2,\dots,j\}$ have never been modified; in particular, we have not moved any objects from levels $\{\ell-1,\ell-2,\dots,j\}$ to levels $\{j-1,j-2,\dots,0\}$. Thus, the increase $s(i,j)-s'(i,j)$ can only be caused by the at most $(c_i\bmod 2^{j-i})$ objects of scale $(2^i,2^{i+1}]$ inserted after time $t'$. Thus, $s(i,j) \le s'(i,j) + (c_i\bmod 2^{j-i})\cdot 2^{i+1}$, finishing the proof of \cref{item:inva-nottoobig}.
\end{proof}

\begin{proof}[Proof of \cref{item:inva-supply}]
To prove the claimed bound on $v(\tau,j)$ at any time step $t$, suppose the most recent execution of $\rebu(j^*)$ such that $j^*\ge j$ finished at time $t'\le t$; if no such execution of $\rebu(j^*)$ has  occurred, let $t'=0$.
Recall from the assumptions that type-$\tau$ bundles have objects of scale $(2^{i}, 2^{i+1}]$, where $i<j$.
As in the proof of \cref{item:inva-nottoobig}, we know that after time $t'$ there have been at most $(c_i \bmod 2^{j-i})$ updates of object scale $(2^i,2^{i+1}]$.

After time $t'$, the levels $\{\ell-1,\ell-2,\dots,j\}$ have never been modified. Hence, by our assumption that $\lowlevel(\tau)\ge j$ holds at time $t$, we know $\lowlevel(\tau)\ge j$ also held at time $t'$. In particular, this means $t'\ge 1$.

Let $v'(\tau,j)$ denote the number of type-$\tau$ bundles in levels $\{j-1,j-2,\dots,0\}$ at time $t'$.
We now show that $v'(\tau,j)\ge 2^{j-i}$. To show this, focus on the execution of $\rebu(j^*)$ at time $t'$, and separately analyze two cases depending on whether $\tau$ is a new type:
\begin{itemize}
   \item Case $\tau \notin \caT_{\mathrm{old}}$ (i.e., $\tau$ is a new type created by $\crea$):
   
   Suppose the \textbf{for} loop at \cref{line:putbundlesfor,line:putbundles} put all the type-$\tau$ bundles in the levels $\{j_0,j_0-1,\dots,i\}$ but not in level $j_0+1$. Then, $ \lowlevel(\tau) = j_0  \in [i,j^*]$.
   Since we assumed $j\le \lowlevel(\tau) = j_0$, we must have $v'(\tau,j) = n_i + n_{i+1}+ \dots + n_{j-1} =  2^{j-i}$. %

      \item Case $\tau \in \caT_{\mathrm{old}}$:
      
      We know $\lowlevel(\tau)\ge j^*+1$ by the definition of $\caT_{\mathrm{old}}$. 
By induction, we can assume as an inductive hypothesis that \cref{item:inva-supply} held at time $t'$ for $v'(\tau,\cdot )$ whenever the second argument is strictly larger than $j$.
In particular,  since $\lowlevel(\tau) \ge j^*+1>j$, this implies \[v'(\tau,j^*+1) \ge 2^{j^*+1-i}-(c_i' \bmod 2^{j^*+1-i}),\]
where $c_i'$ is the counter at time $t'$. By definition of $j^*$ at time $t'$, we know $c'_i$ is divisible by $2^{j^*-i}$ but not by $2^{j^*-i+1}$, so $c_i' \bmod 2^{j^*+1-i} = 2^{j^*-i}$. 
Hence $v'(\tau,j^*+1) \ge 2^{j^*+1-i}-2^{j^*-i} = 2^{j^*-i}$; in other words, in the execution of $\rebu(j^*)$, the number of old type-$\tau$ bundles in levels $\{j^*,j^*-1,\dots,0\}$ is at least $|\caB_{\tau}| \ge 2^{j^*-i}$. 
          Thus, by the placement rule at the end of $\rebu(j^*)$, the number of bundles assigned to levels $\{j-1,j-2,\dots,i\}$ is 
          $v'(\tau,j) = \min\{|\caB_{\tau}|,n_i+n_{i+1}+\dots+n_{j-1}\} = \min\{|\caB_{\tau}|, 2^{j-i}\} \ge \min\{2^{j^*-i}, 2^{j-i}\} = 2^{j-i}$. %
\end{itemize}
Thus, in both cases we have $v'(\tau,j)\ge 2^{j-i}$ as claimed.

After time $t'$, the levels $\{\ell-1,\ell-2,\dots,j\}$ have never been modified; in particular, we have not moved any objects from levels $\{j-1,j-2,\dots,0\}$ to levels $\{\ell-1,\ell-2,\dots,j\}$. Thus, the decrease $v'(\tau,j)-v(\tau,j)$ can only be caused by the at most $(c_i\bmod 2^{j-i})$ objects of scale $(2^i,2^{i+1}]$ deleted after time $t'$, each of which may destroy one type-$\tau$ bundle. 
Hence, $v(\tau,j) \ge v'(\tau,j) - (c_i\bmod 2^{j-i})$, finishing the proof of  \cref{item:inva-supply}. 
\end{proof}

Therefore, our allocator is correct.
Finally, we bound the worst-case expected overhead.
\begin{proof}[Proof of \cref{lem:largemain}]
By definition, $\rebu(j^*)$ only moves objects in levels $\{j^*,j^*-1,\dots,0\}$, whose total size can be bounded by summing up \cref{item:inva-nottoobig} for $j=j^*+1$ over all scales $(2^i,2^{i+1}]$ where $0\le i \le j^*$, giving the upper bound
$ O(2^{j^*} \ell^2(\log\eps^{-1})(\log \ell + \log\log \eps^{-1})^2)$.   
The total size of objects moved by $\inse(\cdot)$ or $\dele(\cdot)$ is dominated by that of $\rebu(j^*)$.
   
  For every update of scale $(2^{i^*},2^{i^*+1}]$ in the update sequence,
     since we randomly initialized the counter $c_{i^*}$, the probability that it triggers $\rebu(j^*)$ equals 
    $\begin{cases}2^{i^*-j^*-1} & i^*\le j^*<\ell-1 \\ 2^{i^*-\ell+1} & j^*=\ell-1\end{cases} \le 2^{i^*-j^*}$.
   Hence, the expected switching cost for this update is 
   \begin{align*}
&   \sum_{i^*\le j^*\le \ell-1}2^{i^*-j^*} \cdot   O(2^{j^*} \ell^2(\log\eps^{-1})(\log \ell + \log\log \eps^{-1})^2)\\
&= O(2^{i^*}\ell^3(\log\eps^{-1}) (\log \ell + \log \log \eps^{-1})^2).
   \end{align*}
 Dividing by the size $\Theta(2^{i^*})$ of the updated object gives the expected overhead bound $O(\ell^3(\log\eps^{-1}) (\log \ell + \log \log \eps^{-1})^2)$  as claimed.
\end{proof}

\section{Lower bounds}
\label{sec:lb}

In this section, we prove \cref{thm:loglowerbound} and \cref{thm:mainsecondmoment}.
Both our lower bounds are proved with an oblivious adversary, i.e., our hard update sequences are fixed before starting the allocator, instead of being generated adaptively based on the current memory state.

\paragraph*{Bounded ratio property.}
The object sizes in our hard update sequences are always in the interval $[\mu,C\mu]$ for some $\mu>0$ and an absolute constant $C> 1$. 
This allows us to work with the more convenient notion of (unnormalized) \emph{switching cost}, namely the total size of changed (moved/inserted/deleted) objects for handling an update $x$, and only lose a constant factor~$C$ when normalized by the size of $x$ to get the overhead factor.

\paragraph*{Full-memory assumption.}
In our proofs, we design hard instances in the original discrete setting (instead of the real-interval setting used in \cref{sec:ds}) with a memory of $M$ slots indexed by $0,1,\dots,M-1$. 
Moreover, by \cref{obs:makememoryfull}, \cref{obs:item2}, we are allowed to make \emph{all the $M$ memory slots fully occupied} in the hard instances.
Formally, we actually prove the following two theorems which imply \cref{thm:loglowerbound} and \cref{thm:mainsecondmoment}. Here, all hidden constants are absolute.

\begin{theorem}
   \label{thm:loglowerbound-discrete}
     For every $M_0 \in \Z^+$, there exists an integer $M\in \Theta(M_0)$ such that any allocator that handles instances with $M$ memory slots
     and object sizes $\Theta(M/\log M)$ must have  worst-case expected switching cost $\Omega (M)$ (against an oblivious adversary).
\end{theorem}

\begin{theorem}
   \label{thm:squaredlowerbound-discrete}
     For every $M_0 \in \Z^+$, there exists an integer $M\in \Theta(M_0)$ such that any allocator that handles instances with $M$ memory slots and object sizes $\Theta(M^{4/7})$ must have worst-case expected squared switching cost $\Omega (M^{9/7})$ (against an oblivious adversary).
\end{theorem}

We suspect that the theorems above could be strengthened to hold for \emph{all} $M\ge M_0$, but our current proofs do not seem to directly imply that.

We can immediately derive \cref{thm:loglowerbound} from \cref{thm:loglowerbound-discrete} as follows.
\begin{proof}[Proof of \cref{thm:loglowerbound} assuming \cref{thm:loglowerbound-discrete}]
 Given $\eps>0$ sufficiently small, pick an integer $M\in \Theta(1/\eps)$ so that $M \le 1/(3\eps)$ and \cref{thm:loglowerbound-discrete} holds for $M$.
 Since the object size in
 \cref{thm:loglowerbound-discrete}
 is $\Theta(M/\log M)$, the required expected overhead factor is $\Omega(M)/\Theta(M/\log M) = \Omega(\log M) = \Omega(\log\eps^{-1})$.
 By \cref{obs:makememoryfull}, \cref{obs:item2} (which can apply to expected overhead as well), we conclude that any allocator that supports load factor $1-\eps \ge 1-\frac{1}{3M}$ must also incur expected overhead at least~$\Omega(\log \eps^{-1})$.
\end{proof}
Analogously, \cref{thm:squaredlowerbound-discrete} implies \cref{thm:mainsecondmoment} (we omit the details).

\cref{thm:loglowerbound-discrete} and \cref{thm:squaredlowerbound-discrete} are proved in \cref{sec:loglb} and \cref{sec:polylb} respectively.

\subsection{Basic definitions}
\label{subsec:lbdefs}
We start by introducing a few definitions which will be used in both \cref{sec:loglb} and \cref{sec:polylb}.

\paragraph*{Size profile $S$ and difference $\delta(S_1,S_2)$.}
In order to design hard input sequences of updates, we actually design hard sequences of \emph{full-memory size profiles} (which could then be translated into sequences of updates).

A \emph{full-memory size profile} (or \emph{size profile} for short) is a multiset $S$ of positive integers whose sum equals $M$, which correspond to the sizes of all live objects at a certain time step. (The labels of the objects are unimportant, since we do not have to distinguish two objects of the same size.)
Define the \emph{difference} between two multisets $S_1,S_2$ as $\delta(S_1,S_2) \coloneqq \sum_{x} |m_{S_1}(x)-m_{S_2}(x)|$, where $m_S(x)$ denotes the multiplicity of $x$ in the multiset $S$.

A sequence of size profiles $(S_1,S_2,\dots,S_k)$ naturally induces a shortest update sequence that realizes them in order (starting from an empty memory), whose length is given by $\delta(\varnothing,S_1)+\sum_{i=1}^{k-1} \delta(S_i,S_{i+1})$.\footnote{If there are multiple such shortest update sequences, we fix an arbitrary one of them.}

As an example, for two size profiles $S_1=\{2,3\}, S_2=\{1,2,2\}$ with memory size $M=5$, the sequence $(S_1,S_2)$ can be realized by the following shortest update sequence starting from an empty memory, whose length is $\delta(\varnothing,S_1)+\delta(S_1,S_2) = 2+ 3 = 5$: 
\begin{itemize}
   \item Insert $x$ with $\mu(x)=2$.
   \item Insert $y$ with $\mu(y)=3$ (realizing size profile $S_1$).
   \item Delete $y$.
 \item Insert $z$ with $\mu(z)=1$.
 \item Insert $w$ with $\mu(w) = 2$ (realizing size profile $S_2$).
\end{itemize}

\paragraph*{Memory state $\phi$, difference $\Delta(\phi,\phi')$, and maximal changed intervals.}
To analyze the necessary cost incurred by the allocator, we are primarily interested in the memory states (i.e., allocations) made by the allocator when the memory is full.

At full memory, the \emph{memory state} $\phi$ can be uniquely described by the list of live object sizes ordered by their locations from left to right.
For example, in the full memory state $\phi = (2,5,3,1,3)$ (with $M = 14$ and size profile $\{1,2,3,3,5\}$), the size-$5$ object has location $2$, i.e., it is allocated to the memory slots with indices in the interval $[2,2+5)$. See visualization in \cref{fig:examplephi}.

We now define the \emph{difference} between two full memory states, which is a lower bound on the total switching cost required to transform one state to the other.

Recall that a size-$\mu$ object is said to have \emph{location} $p$ if it is allocated to the interval $[p,p+\mu)$.

\begin{definition}[$\Delta(\phi,\phi')$]
   \label{defn:delta}
Given two full memory states $\phi,\phi'$, we say a size-$\mu$ object at location $p$  in $\phi$ is \emph{unchanged}, if $\phi'$ also has a size-$\mu$ object at location $p$; otherwise, we say the object is \emph{changed}.
Then, define the \emph{difference} $\Delta(\phi,\phi')$ as the total size of changed objects in $\phi$.
\end{definition}

\begin{observation}
   \label{obs:diff-lb-switchcost}
For two full memory states $\phi,\phi'$, the total switching cost to transform $\phi$ to $\phi'$ is at least $\Delta(\phi,\phi')$.
\end{observation}

We now define the useful notion of \emph{maximal changed intervals}, which we borrow from the previous work \cite{Farach-ColtonKS24}. 
\begin{definition}[Maximal changed intervals]
   \label{defn:maximalchanged}
In the same setup as \cref{defn:delta},
we say $[L,R)\subseteq [0,M)$ is a \emph{maximal changed interval} if it is an inclusion-maximal interval such that every object in $\phi$ intersecting this interval is changed. 

By definition, $[0,M)$ is partitioned into the disjoint union of the unchanged objects and the maximal changed intervals. 
In particular, $\Delta(\phi,\phi')$ equals the total length of the maximal changed intervals.
\end{definition}

We remark that the definitions of $\Delta(\phi,\phi')$, unchanged objects, and maximal changed intervals are symmetric with respect to $\phi$ and $\phi'$.

See \cref{fig:examplephi} for an example with visualization that illustrates the definitions above.

\subsection{Logarithmic lower bound for expected overhead}
\label{sec:loglb}

In this section, we prove \cref{thm:loglowerbound-discrete} (which implies \cref{thm:loglowerbound}).

We will define a hard distribution over sequences of full-memory size profiles. 
We first specify the memory size $M$, and the family of size profiles that we will use.

Let parameter $k$ be a positive integer to be determined later. 

\begin{definition}[Size profile $S_\pi$]
For a permutation $\pi\colon [k]\to [k]$, define the size profile \[S_{\pi} = \{2^{2k+1} + 2^{k+i} + 2^{\pi(i)}: i\in \{1,2,\dots,k\}\},\]
which consists of $k$ distinct-sized objects of total size 
\[M:= k\cdot 2^{2k+1} + \sum_{i=1}^{k}2^{k+i} + \sum_{i=1}^k 2^{i} = (k+1)2^{2k+1} - 2  .\] %

Note that all object sizes are in the interval $[2^{2k+1},2^{2k+2})$.
\end{definition}
By the definition above, the number of memory slots is $M \in \Theta(k2^{2k}) $.
Hence, given $M_0\ge 1$, we can choose some $k \in \frac{1}{2}(\log M_0 - \log\log M_0) + O(1)$ so that $M \in \Theta(M_0)$. Then, the object sizes are in $[\mu,2\mu)$ where $\mu \in \Theta(M/\log M)$.
Recall that our goal is to prove a worst-case expected switching cost lower bound of $\Omega(M)$ in this full-memory setting.

We have the following simple but crucial lemma:
\begin{lemma}
   \label{lem:binary}
  For two permutations $\pi,\pi'\colon [k]\to [k]$, suppose two non-empty subsets $X\subseteq S_{\pi}, X'\subseteq S_{\pi'}$ have equal sum $\sum_{x\in X} x = \sum_{x'\in X'} x'$. 
 Then, there is a set $J \subseteq [k]$ such that 
 \begin{align}
  X &= \{2^{2k+1}+2^{k+i}+2^{\pi(i)} : i\in J\},\label{eqn:tempx}  \\ X' &= \{2^{2k+1}+2^{k+i}+2^{\pi'(i)} : i\in J\}, \label{eqn:tempxprime}
 \end{align}
   and 
   \begin{equation}
    \{\pi(i):i\in J\}=\{\pi'(i):i\in J\}. \label{eqn:piJ}
   \end{equation}
  \end{lemma}
\begin{proof}
  Consider the binary representation of $\sum_{x\in X} x$, restricted to the part between the $(k+1)$-st bit and the $2k$-th bit (inclusive; the $i$-th bit has binary weight $2^i$).
  By definition of $S_\pi$,  the positions of 1s among the bits in this part uniquely determine the set $J\subseteq [k]$
 that satisfies \cref{eqn:tempx}, 
 since there are no carries from the lower binary bits in the summation.
  Since this sum is the same as $\sum_{x'\in X'} x'$, we know that the same $J$ should also satisfy \cref{eqn:tempxprime}.
  
  Similarly, the part 
between the $1$-st bit and the $k$-th bit in the binary representation of
$\sum_{x\in X} x$ can uniquely determine the set $\{\pi(i): i\in J\}$. 
Since this sum is the same as $\sum_{x'\in X'} x'$, we conclude that \cref{eqn:piJ} must hold.
\end{proof}

Now we are ready to prove \cref{thm:loglowerbound-discrete}.
\begin{proof}[Proof of \cref{thm:loglowerbound-discrete}]
We define a distribution over length-two sequences of size profiles, $(S_{\iota},S_{\pi'})$,
where $\iota \colon [k] \to [k]$ is the identity permutation $\iota(i)=i$, and  
let $\pi'\colon [k] \to [k]$ be modified from $\iota$ by swapping  two random coordinates, that is,  we independently pick two uniformly random \emph{distinct} indices $a,b\in [k]$ and
define 
\[ \pi'(i)\coloneqq \begin{cases}
  a & \text{if }i=b,\\
  b & \text{if }i=a,\\
  i & \text{otherwise.}\\
\end{cases}\]
Note that $\delta(S_{\iota},S_{\pi'})=4$, since the adversary can transform $S_\iota$ to $S_{\pi'}$ by first deleting the two objects of sizes 
\[x\coloneqq 2^{2k+1} + 2^{k+a} + 2^{a}, y\coloneqq 2^{2k+1} + 2^{k+b} + 2^{b}\], and then inserting another two objects of sizes
\[x'\coloneqq 2^{2k+1} + 2^{k+b} + 2^{a}, y'\coloneqq 2^{2k+1} + 2^{k+a} + 2^{b}.\]
Initially, the size profile $S_{\iota}$ can be realized by inserting $k$ objects into the empty memory.

We run the randomized allocator on the random update sequence induced by the sequence of size profiles $(S_\iota,S_{\pi'})$. 
We compare the full memory states $\phi_{\iota}$ and $\phi_{\pi'}$ corresponding to the size profiles $S_{\iota}$ and $S_{\pi'}$ respectively.
\begin{claim}
There is a maximal changed interval between $\phi_{\iota}$ and $\phi_{\pi'}$ that contains both the size-$x$ and size-$y$ objects of $\phi_{\iota}$.
\end{claim}
\begin{proof}
The size-$x$ and size-$y$ objects appear in $\phi_{\iota}$ but no longer appear in $\phi_{\pi'}$, so both of them are changed objects, and each should be in some maximal changed interval (\cref{defn:maximalchanged}).
Let $[L,R)$ be the maximal changed interval that contains the size-$x$ object in $\phi_{\iota}$.
This interval consists of a subset of changed objects in $\phi_{\iota}$ including the size-$x$ object. Hence,
 the interval length $R-L$ equals the sum of some subset $X\subseteq S_{\iota}$, where $X \ni x$. 
Similarly, this interval length $R-L$ also 
equals the sum of some subset $X' \subseteq S_{\pi'}$.
Hence, $\sum_{x\in X} x = \sum_{x'\in X'} x'$.

We can now apply \cref{lem:binary} to $\iota,\pi'$ and $X,X'$, and obtain $J\subseteq [k]$. 
From $x \in X$ and \cref{eqn:tempx}, we get $a\in J$; in particular, $\iota(a) \in \{\iota(i) : i \in J\}$. 
Then, by \cref{eqn:piJ}, we have $\iota(a) \in \{\pi'(i) : i \in J\}$ as well.
Since $\iota(a) = a = \pi'(b)$,  this implies $\pi'(b)\in \{\pi'(i) : i \in J\}$, and thus $  b\in J$.
By $b\in J$ and \cref{eqn:tempx}, we get $y\in X$.
Therefore, the size-$y$ object is also contained in the interval $[L,R)$, which proves the claim.
\end{proof}

By the claim above, all objects in $\phi_\iota$ located 
between the size-$x$ and size-$y$ objects (inclusive)
are changed objects, whose sizes should contribute to $\Delta(\phi_{\iota},\phi_{\pi'})$. 
Since $a\neq b \in [k]$ are uniformly randomly chosen by the adversary and are unknown to the allocator in advance, one can show that the expected number of objects in $\phi_{\iota}$ located between the size-$x$ and size-$y$ objects (inclusive) equals $\frac{k+4}{3}$. 
Since all object sizes are in $[\mu,2\mu] = [2^{2k+1},2^{2k+2}]$, we get that $\Delta(\phi_{\iota},\phi_{\pi'})$ has expectation at least $\frac{k+4}{3} \mu$.
Hence, by \cref{obs:diff-lb-switchcost}, the expected total switching cost to transform $\phi_{\iota}$ to $\phi_{\pi'}$ is at least
$\frac{k+4}{3} \mu$.
Since this transformation is completed within $\delta(S_{\iota},S_{\pi'})=4$ updates,
we know the worst-case expected switching cost achieved by the allocator must be at least $\frac{k+4}{3} \mu/ \delta(S_{\iota},S_{\pi'}) \ge \Omega(M)$. 
This finishes the proof of \cref{thm:loglowerbound-discrete}.
\end{proof}

\subsection{Polynomial lower bound for expected squared overhead}
\label{sec:polylb}
In this section, we prove \cref{thm:squaredlowerbound-discrete} (which implies \cref{thm:mainsecondmoment}).
We described the intuition in the overview section (\cref{subsec:overview}). Now
we outline the structure of the formal proof:
\begin{itemize}
   \item  In \cref{subsub:sizeprof}, we construct the family $\{S_{i,j}\}$ of size profiles used in our proof. We establish the property of finger objects as mentioned in the overview section.
   \item  In \cref{subsub:tree}, we use $\{S_{i,j}\}$ to design hard update sequences. 
   These sequences are organized as a tree $\caT$, where each node corresponds to an $S_{i,j}$.
   The tree $\caT$ has a spine and many branches. 
    The  ``surprise inspection'' mentioned in the overview intuitively corresponds to
   branching off from a random point on the spine.
   \item  In \cref{subsub:yao}, assuming a good randomized allocator $\caA^{\mathrm{rand}}$ exists, we use Yao's principle to fix a certain good deterministic allocator $\caA$. Then, each node on the tree $\caT$ is associated with a memory state, namely the state reached by running $\caA$ on the update sequence induced by the root-to-node path on $\caT$. 
      From the performance of $\caA$, we obtain certain inequalities involving the differences between these memory states.
   \item  In \cref{subsub:mainproof}, we use the properties of finger objects to show that these inequalities cannot hold, reaching a contradiction.  Hence the purported allocator $\caA^{\mathrm{rand}}$ cannot exist, finishing the proof.
\end{itemize}

\subsubsection{Size profiles and their properties}
\label{subsub:sizeprof}
We construct the family of size profiles that we will use in our hard instances, and prove its main property in \cref{lem:finger}.  
The construction is based on two integer parameters $P> Q\ge 1$, which will be determined later. 

We need a standard construction of several integers whose  small-coefficient linear combinations are all distinct:
\begin{lemma}
   \label{lem:fournumbers}
  Given $n_1\ge n_2 \ge n_3 \ge n_4\ge 1$, there exist positive integers $a_1,a_2,a_3,a_4 \in \Theta(n_2n_3n_4)$, such that the integers $k_1a_1+k_2a_2+k_3a_3 + k_4a_4$ for $k_i \in \Z \cap [-n_i,n_i]$ are all distinct.
  Moreover, $a_1< a_2< a_3< a_4 \le 2a_1$.
\end{lemma}
\begin{proof}
   Let $n_i' \coloneqq  2n_i+1$ for all $i\in [4]$. Let
   \begin{align*}
   a_1 & \coloneqq n_2'n_3'n_4'\\
   a_2 & \coloneqq a_1 + n_3'n_4'\\
   a_3 & \coloneqq a_2 + n_4'\\
   a_4 & \coloneqq a_3 + 1,
   \end{align*}
   which clearly satisfy the ``moreover'' part of the lemma statement.

  Suppose to the contrary that 
$k_1a_1+\dots + k_4a_4 =k'_1a_1+\dots + k'_4a_4$ for some $(k_1,\dots,k_4)\neq (k'_1,\dots,k'_4)$ where $k_i,k'_i \in \Z\cap [-n_i,n_i]$. Then, 
$m_1a_1+\dots + m_4a_4 = 0 $, where $m_i = k_i-k_i' \in [-2n_i,2n_i]$ and $m_i$ are not all zero. 
Since $a_1\equiv a_2\equiv a_3 \equiv 0 \pmod{n'_4}$ and $a_4\equiv 1\pmod{n'_4} $,  we obtain $ 0 =  m_1a_1+\dots +m_4a_4\equiv m_4 \pmod{n'_4} $. Since $|m_4| \le 2n_4 = n'_4-1$, we must have $m_4=0$.

Then, $m_1a_1+m_2a_2+m_3a_3=0$. We divide both sides by $n'_4$, and note that $\frac{a_1}{n'_4} \equiv \frac{a_2}{n'_4}\equiv 0\pmod{n'_3}$, $\frac{a_3}{n'_4}\equiv 1 \pmod{n'_3}$. Using the same argument as the previous paragraph, we get $m_3=0$. Iterating the argument again gives $m_2=0$, at which point we conclude $m_1,\dots,m_4$ must all be zero, a contradiction. Thus $a_1,a_2,a_3,a_4$ satisfy the desired conditions.
\end{proof}

Let $P> Q\ge 1$ be integer parameters to be determined later. 

Use \cref{lem:fournumbers} to construct four positive integers $a,b,c,d$ such that
the integers
\begin{equation}
   \label{eqn:unique}
  k_1a+k_2b+k_3c+k_4d \ :\  k_1,k_2 \in \Z\cap [-3P,3P], k_3 \in \Z \cap [0,2Q],k_4\in \{0,1\}
\end{equation}
are all distinct, and 
\begin{equation}
   \label{eqn:rangesize}
 \mu \coloneqq a<b<c<d \le 2\mu
\end{equation}
 where $\mu \in \Theta(PQ)$.
We define the number of memory slots in our hard instance to be
\begin{equation}
  M\coloneqq Pa + Qc + d \in [P\mu, 3P\mu]  \subset \Theta(P^2Q).
  \label{eq:total-memory}
\end{equation}

\begin{definition}[Size profile $S_{i,j}$]
   \label{defn:sizeprof}
   For integers $0\le i\le P, 0\le j\le Q$, define the full-memory size profile $S_{i,j}$ as the multiset consisting of:
   \begin{itemize}
      \item     $i$ size-$a$ objects,
         \item  $j$ size-$c$ objects,
\item  $h$ size-$b$ objects, where $h\coloneqq h(i,j)\coloneqq \lfloor (M-d-ia-jc)/{b}\rfloor$ (by \cref{eq:total-memory}, $0\le h\le 3P$ must hold), and
 \item one remaining object (termed the \emph{finger object}) of size $M - ia - jc - hb$.
   \end{itemize}
\end{definition}
The finger object size $f$ in \cref{defn:sizeprof} satisfies
\[f-d  = (M - d- ia - jc) - \lfloor (M-d-ia-jc)/{b}\rfloor \cdot b\in [0,b),\] that is, $f\in [d,d+b)$. Combining this bound with \cref{eqn:rangesize} yields the following basic observations:
\begin{observation}
   \label{obs:boundedratio}
 Every object in $S_{i,j}$ has size in $[\mu,4\mu]$.
\end{observation}
\begin{observation}
   The finger object in $S_{i,j}$ has size different from $a$, $b$, and $c$ (which justifies its distinguished role).
\end{observation}

We now describe the arithmetic structure of subset sums of object sizes in $S_{i,j}$:
\begin{lemma}
   \label{cor:sizestructure}
Let $X$ be any subset of the multiset $S_{i,j}$. Then, there exists a unique tuple 
$(k_1,k_2,k_3,k_4)$ in the range $k_1\in \Z\cap [0,2P], k_2 \in \Z\cap [-3P,3P], k_3\in \Z\cap [0,2Q], k_4\in \{0,1\}$, such that the sum of $X$ equals $k_1a+k_2b+k_3c+k_4d$. 
Moreover, $k_4=1$ if and only if $X$ contains the finger object of $S_{i,j}$.
\end{lemma}
\begin{proof}
Given $X\subseteq S_{i,j}$, we first show that there exists \emph{some} $(k_1,k_2,k_3,k_4)$ from the specified range such that 
the sum of $X$ equals $k_1a+k_2b+k_3c+k_4d$.

Recall from \cref{defn:sizeprof} that the numbers of size-$a$, size-$b$, and size-$c$ objects in $S_{i,j}$ are $i\le P$, $h\le 3P$, and $j\le Q$ respectively.
Denote the number of size-$a$, size-$b$, and size-$c$ objects contained in $X \subseteq S_{i,j}$ by $m_1\in [0,P],m_2\in [0,3P],$ and $m_3\in [0,Q]$, respectively.  
Now consider two cases:
\begin{itemize}
   \item If $X$ does not contain the finger object of $S_{i,j}$, then the sum of $X$ equals $m_1a+m_2b+m_3c$, so $(k_1,k_2,k_3,k_4)\coloneqq (m_1,m_2,m_3,0)$ satisfies the requirement and is in the specified range.
      \item Otherwise, $X$ contains the finger object of size $f$. By \cref{defn:sizeprof} and \cref{eq:total-memory}, $f = M-ia-jc-hb = (P-i)a-hb+(Q-j)c+d$. 
      Thus, the sum of $X$ equals $f + m_1a+m_2b+m_3c = (P-i+m_1)a + (m_2-h)b + (Q-j+m_3)c + d$, so $(k_1,k_2,k_3,k_4)\coloneqq (P-i+m_1,m_2-h,Q-j+m_3,1)$ 
      satisfies the requirement and is in the specified range.
\end{itemize}

Hence, the desired $(k_1,k_2,k_3,k_4)$ in the specified range always exists. The uniqueness of this tuple in this range then follows from the fact that all integers in \cref{eqn:unique} are distinct.

The ``moreover'' part of the lemma statement follows from the case distinction above and the uniqueness of the tuple $(k_1,k_2,k_3,k_4)$.
\end{proof}

\cref{cor:sizestructure} implies the following crucial lemma, which is key to our proof:
\begin{lemma}
   \label{lem:finger}
  Let  $\phi,\phi'$ be two memory states corresponding to two size profiles $S_{i,j},S_{i',j'}$ respectively.
 Let $[L,R)$ be a maximal changed interval between $\phi,\phi'$. 
 
 If $[L,R)$ in $\phi$ does not contain the finger object of $S_{i,j}$, then the multiset of object sizes in $[L,R)$ in $\phi'$ is the same as that in $\phi$ (in particular, $[L,R)$ in $\phi'$ does not contain the finger object of $S_{i',j'}$).
\end{lemma}

See \cref{fig:finger} for an example. 
By taking the contrapositive and using the symmetry between $\phi$ and $\phi'$ in the statement above, we have the following corollary:  $[L,R)$ in $\phi$ contains the finger object of $S_{i,j}$ if and only if  $[L,R)$ in $\phi'$ contains the finger object of $S_{i',j'}$.

\begin{proof}[Proof of \cref{lem:finger}]
Let $X \subseteq S_{i,j}$ and 
$X' \subseteq S_{i',j'}$ denote the multisets of sizes of the objects contained in $[L,R)$ in $\phi$ and $\phi'$ respectively.
Then the sum of $X$ and the sum of $X'$ both equal $R-L$.

Suppose $[L,R)$ in $\phi$ does not contain the finger object. Then, $R-L=k_1a+k_2b+k_3c+k_4d$, where $k_4=0$, and  $k_1\le P,k_2 \le 3P,k_3\le Q$ denote the number of size-$a$, size-$b$, size-$c$ objects in $X$ respectively.
By the ``moreover'' part of \cref{cor:sizestructure} applied to $X'$ and $k_4=0$, we conclude $X'$ does not contain the finger object of $S_{i',j'}$. 
 Therefore, $R-L =  k_1'a+k_2'b+k_3'c+0\cdot d$ where
 $k_1'\le P,k_2' \le 3P,k_3'\le Q$ denote the number of size-$a$, size-$b$, size-$c$ objects in $X'$ respectively.
 By the uniqueness in \cref{cor:sizestructure}, we must have $k_i=k_i'$ for all $i\in [3]$, i.e.,   the multisets $X$ and $X'$ are equal.
\end{proof}

\subsubsection{Hard update sequences}
\label{subsub:tree}

We now proceed to the construction of hard sequences of size profiles.

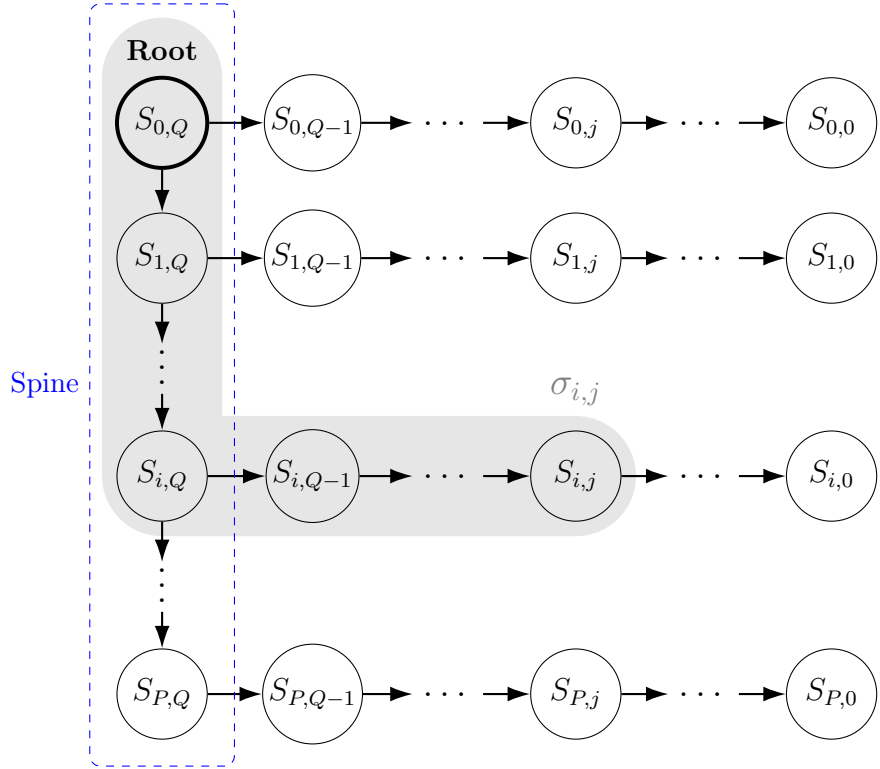
\begin{figure}[htbp]
\centering

\begin{tikzpicture}[
    node_style/.style={
        draw, 
        circle, 
        minimum size=1.2cm, 
        inner sep=1pt, 
        font=\large
    },
    dots_style/.style={
        draw=none, 
        rectangle, 
        font=\Large
    },
    edge_style/.style={
        draw, 
        ->, 
        >={Latex[length=3mm, width=2mm]}, 
        thick
    }
]

\def\xQ{0}           %
\def\xQm{2}        %
\def\xDa{3.8}        %
\def\xj{5.5}         %
\def\xDb{7.2}        %
\def\xZ{8.9}         %

\def\yZero{0}        %
\def\yOne{-1.8}      %
\def\yDa{-3.1}       %
\def\yI{-4.7}        %
\def\yDb{-6.0}       %
\def\yP{-7.6}        %

\node[node_style, line width=1.5pt] (S0Q) at (\xQ, \yZero) {$S_{0,Q}$};
\node[above=0.07cm of S0Q, font=\bfseries] (sroot) {Root};
\node[node_style] (S0Q-1) at (\xQm, \yZero) {$S_{0,Q-1}$};
\node[dots_style] (d0a)   at (\xDa, \yZero) {$\dots$};
\node[node_style] (S0j)   at (\xj,  \yZero) {$S_{0,j}$};
\node[dots_style] (d0b)   at (\xDb, \yZero) {$\dots$};
\node[node_style] (S00)   at (\xZ,  \yZero) {$S_{0,0}$};

\node[node_style] (S1Q)   at (\xQ, \yOne) {$S_{1,Q}$};
\node[node_style] (S1Q-1) at (\xQm, \yOne) {$S_{1,Q-1}$};
\node[dots_style] (d1a)   at (\xDa, \yOne) {$\dots$};
\node[node_style] (S1j)   at (\xj,  \yOne) {$S_{1,j}$};
\node[dots_style] (d1b)   at (\xDb, \yOne) {$\dots$};
\node[node_style] (S10)   at (\xZ,  \yOne) {$S_{1,0}$};

\node[dots_style] (vdotQ) at (\xQ, \yDa) {$\vdots$};

\node[node_style] (SiQ)   at (\xQ, \yI) {$S_{i,Q}$};
\node[node_style] (SiQ-1) at (\xQm, \yI) {$S_{i,Q-1}$};
\node[dots_style] (dia)   at (\xDa, \yI) {$\dots$};
\node[node_style] (Sij)   at (\xj,  \yI) {$S_{i,j}$}; 
\node[dots_style] (dib)   at (\xDb, \yI) {$\dots$};
\node[node_style] (Si0)   at (\xZ,  \yI) {$S_{i,0}$};

\node[dots_style] (vdotQ2) at (\xQ, \yDb) {$\vdots$};

\node[node_style] (SPQ)   at (\xQ, \yP) {$S_{P,Q}$};
\node[node_style] (SPQ-1) at (\xQm, \yP) {$S_{P,Q-1}$};
\node[dots_style] (dPa)   at (\xDa, \yP) {$\dots$};
\node[node_style] (SPj)   at (\xj,  \yP) {$S_{P,j}$};
\node[dots_style] (dPb)   at (\xDb, \yP) {$\dots$};
\node[node_style] (SP0)   at (\xZ,  \yP) {$S_{P,0}$};

\begin{scope}[every edge/.style={edge_style}]

    \path (S0Q) edge (S1Q);
    \path (S1Q) edge[shorten >=-10pt] (vdotQ); 
    \path (vdotQ) edge (SiQ); 
    
    \path (SiQ) edge[shorten >=-10pt] (vdotQ2);
    \path (vdotQ2) edge (SPQ);

    \path (S0Q) edge (S0Q-1);
    \path (S0Q-1) edge (d0a);
    \path (d0a) edge (S0j);
    \path (S0j) edge (d0b);
    \path (d0b) edge (S00);

    \path (S1Q) edge (S1Q-1);
    \path (S1Q-1) edge (d1a);
    \path (d1a) edge (S1j);
    \path (S1j) edge (d1b);
    \path (d1b) edge (S10);

    \path (SiQ) edge (SiQ-1);
    \path (SiQ-1) edge (dia);
    \path (dia) edge (Sij);
    \path (Sij) edge (dib);
    \path (dib) edge (Si0);

    \path (SPQ) edge (SPQ-1);
    \path (SPQ-1) edge (dPa);
    \path (dPa) edge (SPj);
    \path (SPj) edge (dPb);
    \path (dPb) edge (SP0);
\end{scope}

\node[draw,blue, dashed, rounded corners, fit=(sroot) (SPQ), inner sep=10pt, label=left:\textcolor{blue}{Spine}] {};

\begin{scope}[on background layer]
    \draw[line width=1.6cm, color=gray!20, line cap=round, line join=round] 
        ([yshift=6mm]S0Q.center) -- (SiQ.center) -- (Sij.center);
\end{scope}

\node[above=0.2cm of Sij, font=\Large] (sigmaij) {\textcolor{gray}{$\sigma_{i,j}$}};
\end{tikzpicture}

\caption{The rooted tree $\caT$ defined in \cref{defn:sigma}. The sequence $\sigma_{i,j}$ (shaded in gray) is the tree path from the root node $S_{0,Q}$ to the node $S_{i,j}$.
The tree has a \emph{spine} (shown in blue dashed box) $(S_{0,Q},S_{1,Q},\dots, S_{P,Q})$, and each node on the spine is the root of a \emph{branch} $(S_{i,Q},S_{i,Q-1},\dots,S_{i,0})$. }
\label{figure:sigmatree}
\end{figure}

  \begin{definition}[Tree $\caT$ and sequence $\sigma_{i,j}$]
     \label{defn:sigma}
    Define a rooted tree $\caT$ with $(P+1)(Q+1)$ nodes uniquely labeled by the size profiles $S_{i,j}$ with $0\le i\le P,0\le j\le Q$ (\cref{defn:sizeprof}), as follows (see an illustration in \cref{figure:sigmatree}): The root node is $S_{0,Q}$. For each $1\le i\le P$, the parent of $S_{i,Q}$ is $S_{i-1,Q}$. For each $0\le i\le P$ and $0\le j\le Q-1$, the parent of $S_{i,j}$ is $S_{i,j+1}$. 

 Define $\sigma_{i,j}$ as the sequence of the size profiles corresponding to the path on $\caT$ from the root to $S_{i,j}$, i.e., 
  \[ \sigma_{i,j} \coloneqq  (S_{0,Q},S_{1,Q},\dots,S_{i,Q}, S_{i,Q-1},S_{i,Q-2},\dots,S_{i,j+1},S_{i,j}). \]
  \end{definition}

  Observe that the difference $\delta(\cdot,\cdot)$ between two adjacent size profiles on $\caT$ is always $O(1)$:
\begin{observation}
   \label{obs:deltaconstant}
For all $0\le i\le P,1\le j\le Q$,  $\delta (S_{i,j},S_{i,j-1})\le 5$. 

For all $1\le i\le P$,  $\delta (S_{i,Q},S_{i-1,Q})\le 5$. 
\end{observation}
  \begin{proof}
We prove the first claim only (the second claim follows from the same argument).
  By \cref{defn:sizeprof}, $S_{i,j}$ has  one more size-$c$ object than $S_{i,j-1}$, and the same number of size-$a$ objects. The difference between the counts of size-$b$ objects is \[0\le h(i,j-1)-h(i,j) =
 \lfloor (M-d-ia-jc+c)/{b}\rfloor - \lfloor (M-d-ia-jc)/{b}\rfloor \le 2,\] 
 where the last step is due to $c/b\le 2\mu/\mu =2$. Finally, $S_{i,j}$ and $S_{i,j-1}$ each have a finger object of possibly different size. Summing these differences together gives $\delta (S_{i,j},S_{i,j-1})\le 1 + 2 + 2 =  5$.
  \end{proof}

\subsubsection{Fixing a deterministic allocator}
\label{subsub:yao}
We will use Yao's principle to find a good deterministic allocator, and then analyze its behavior. 
Before that, we make a few more definitions.

In a consecutive sequence of $\delta$ updates, if the allocator incurs switching costs $C_1,C_2,\dots,C_\delta$, respectively, then we can relate the total squared switching cost to the total switching cost by Cauchy--Schwarz inequality:
\begin{equation}
   \label{eqn:CS}
 (C_1+\dots + C_\delta)^2 \le \delta(C_1^2+\dots +C_\delta^2).
\end{equation}

\begin{definition}[Memory state $\phi_{i,j}$ and total squared switching cost $\kappa(\cdot,\cdot)$]
For a deterministic allocator $\caA$, we use $\phi_{i,j}^{\caA}$ to denote the memory state after running $\caA$ on the update sequence induced by the size profile sequence $\sigma_{i,j}$ (which ends at $S_{i,j}$; see \cref{defn:sigma}).

If $S_{i',j'}$ is the parent of $S_{i,j}$ on the tree $\caT$, then we define $\kappa^{\caA}(\phi^{\caA}_{i',j'},\phi^{\caA}_{i,j})$ to be the total squared switching cost incurred by $\caA$ when transforming state $\phi^{\caA}_{i',j'}$  to state $\phi^{\caA}_{i,j}$.\end{definition}

Using Cauchy--Schwarz inequality (\cref{eqn:CS}) and \cref{obs:diff-lb-switchcost}, we have
\begin{equation}
 \Delta(\phi^{\caA}_{i',j'},\phi^{\caA}_{i,j})^2 \le  \delta(S_{i',j'},S_{i,j}) \cdot  \kappa^{\caA}(\phi^{\caA}_{i',j'},\phi^{\caA}_{i,j}).
 \label{eqn:tempcs}
\end{equation}

Suppose there is a randomized allocator $\caA^{\mathrm{rand}}$ (i.e., a probability distribution over deterministic allocators $\caA$) with worst-case expected squared switching cost at most $F^2\mu^2$ on each update. 
If $S_{i',j'}$ is the parent of $S_{i,j}$ on the tree $\caT$, then running $\caA\sim \caA^{\mathrm{rand}}$ on the deterministic update sequence induced by $\sigma_{i,j}$ yields
\[
 \Ex_{\caA \sim\caA^{\mathrm{rand}} } [\kappa^{\caA}(\phi^{\caA}_{i',j'},\phi^{\caA}_{i,j})] \le \delta(S_{i',j'},S_{i,j})\cdot F^2\mu^2,\] %
which, combined with \cref{eqn:tempcs}, gives
\begin{equation}
 \Ex_{\caA \sim\caA^{\mathrm{rand}} } [\Delta(\phi^{\caA}_{i',j'},\phi^{\caA}_{i,j})^2] \le \delta(S_{i',j'},S_{i,j})^2\cdot F^2\mu^2\le 25F^2\mu^2, 
\label{eqn:expsqeach}
\end{equation}
where the last step is due to \cref{obs:deltaconstant}.

Summing \cref{eqn:expsqeach} along the spine of $\caT$ gives
\begin{equation}
   \label{eqn:randsump}
 \Ex_{\caA \sim\caA^{\mathrm{rand}} } \left [\sum_{i=1}^P\Delta(\phi^{\caA}_{i-1,Q},\phi^{\caA}_{i,Q})^2\right ] \le 25PF^2\mu^2.
\end{equation}
Summing \cref{eqn:expsqeach} over all edges on all the $(P+1)$ branches of $\caT$ gives 
\begin{equation}
   \label{eqn:randsumpq}
 \Ex_{\caA \sim\caA^{\mathrm{rand}} } \left [\sum_{i=0}^P\sum_{j=0}^{Q-1}\Delta(\phi^{\caA}_{i,j+1},\phi^{\caA}_{i,j})^2\right ] \le 25(P+1)QF^2\mu^2.
\end{equation}
By adding \cref{eqn:randsump} and $Q^{-1}$ times \cref{eqn:randsumpq}, applying the linearity of expectation, followed by averaging (Yao's principle), we conclude that there exists a deterministic allocator $\caA$ such that
\[ \sum_{i=1}^P\Delta(\phi^{\caA}_{i-1,Q},\phi^{\caA}_{i,Q})^2 + \frac{1}{Q}\sum_{i=0}^P\sum_{j=0}^{Q-1}\Delta(\phi^{\caA}_{i,j+1},\phi^{\caA}_{i,j})^2 \le 25(2P+1)F^2\mu^2 \le 100PF^2\mu^2.\]
In the following, we fix this deterministic allocator $\caA$, and drop the superscript $\caA$ for simplicity.
In particular, the inequality above implies
\begin{equation}
   \label{eqn:sumP}
  \sum_{i=1}^P \Delta(\phi_{i-1,Q},\phi_{i,Q})^2 \le 100PF^2\mu^2,
\end{equation}
and
\begin{equation}
   \label{eqn:sumPQ}
  \sum_{i=0}^{P} \sum_{j=0}^{Q-1}\Delta(\phi_{i,j+1},\phi_{i,j})^2 \le 100P QF^2\mu^2.
\end{equation}

The rest of our proof amounts to showing that the above two inequalities would lead to a contradiction in a certain parameter regime. We summarize this as the following lemma:
\begin{lemma}
   \label{lem:goal}
  Assume
\begin{align}
   Q^2 & \ge 5\sqrt{P}F,
   \label{eqn:assu1}
   \\
   \sqrt{P} &\ge 800QF.
   \label{eqn:assu3}
\end{align}
Then, there cannot exist memory states $\{\phi_{i,j}\}_{0\le i\le P, 0\le j\le Q}$ such that $\phi_{i,j}$ corresponds to the size profile $S_{i,j}$ and both \cref{eqn:sumP,eqn:sumPQ} hold.
\end{lemma}

We now quickly derive \cref{thm:squaredlowerbound-discrete} from \cref{lem:goal}.
\begin{proof}[Proof of \cref{thm:squaredlowerbound-discrete} assuming \cref{lem:goal}]
Given $M_0\ge 1$, we choose $F = \lceil M_0^{1/14}\rceil$.
Then, we set the parameters $P>Q\ge 1$ (used for the definitions in \cref{subsub:sizeprof,subsub:tree}) as follows to satisfy \cref{eqn:assu1} and \cref{eqn:assu3}:
\begin{align}
   Q &\coloneqq 4000 F^2,
   \label{eqn:assu4}\\
   P &\coloneqq 160 Q^3 = \Theta(F^6).
\label{eqn:assuP}
\end{align}
By \cref{eq:total-memory}, the number of memory slots in our construction is $M = \Theta(P^2 Q) = \Theta(F^{14}) \in \Theta(M_0)$,  and each object has size $\Theta(\mu) = \Theta(PQ)= \Theta(F^8) = \Theta(M^{4/7})$, as required. %

Let $\mathcal{A}^{\mathrm{rand}}$ be a randomized allocator for instances with $M$ memory slots and object sizes $\Theta(M^{4/7})$. If it achieves worst-case expected squared switching cost at most $F^2\mu^2$, then by earlier discussions in \cref{subsub:yao}, there exist memory states $\{\phi_{i,j}\}$ (corresponding to size profiles $S_{i,j}$) such that both \cref{eqn:sumP,eqn:sumPQ} hold, but this is impossible due to \cref{lem:goal}. Hence, $\mathcal{A}^{\mathrm{rand}}$ must have worst-case expected squared switching cost larger than $F^2\mu^2 \ge \Omega(M^{9/7})$, as claimed.
This finishes the proof of \cref{thm:squaredlowerbound-discrete}.
\end{proof}

\subsubsection{Analysis}

\label{subsub:mainproof}

It remains to prove \cref{lem:goal}. We will assume the claimed memory states $\{\phi_{i,j}\}_{0\le i\le P, 0\le j\le Q}$ exist, and derive a contradiction.
In doing so, we will crucially use the property of the finger object (\cref{lem:finger}).

\begin{definition}[$p_{i,j}$]
Let $p_{i,j}$ denote the location of the finger object in the memory state $\phi_{i,j}$.
\end{definition}

We first show that all the size-$c$ objects in the memory state $\phi_{i,Q}$ should be close to the finger object:
\begin{lemma}
   \label{lem:c-clustered}
Define a radius parameter
\begin{equation}
   \label{defn:rho}
   \rho \coloneqq 10\sqrt{P}QF\mu.
\end{equation}
Then,  for all $0\le i \le P$, in the memory state $\phi_{i,Q}$, the interval $[p_{i,Q}-\rho,\ p_{i,Q}+\rho)$ must fully contain all the $Q$ size-$c$ objects.
\end{lemma}
\begin{proof}
Compare the memory states $\phi_{i,Q}$ and $\phi_{i,0}$.
Since $\phi_{i,0}$ contains no size-$c$ objects, we know every size-$c$ object in $\phi_{i,Q}$ must be contained in some maximal changed interval $[L_i,R_i)$ between these two states; moreover, since $[L_i,R_i)$ in $\phi_{i,0}$ 
contains no size-$c$ objects, \cref{lem:finger} implies that $[L_i,R_i)$ must contain the finger objects in both memory states.
In other words, all the size-$c$ objects of $\phi_{i,Q}$ are fully inside the maximal changed interval $[L,R)$ that contains the finger objects. 
 Since $p_{i,Q}\in [L,R)$, %
it then suffices to show $R-L\le \rho$. 

By \cref{defn:maximalchanged}, $R-L \le \Delta(\phi_{i,Q},\phi_{i,0})$.
Then, a very crude application of \cref{eqn:sumPQ} gives
\begin{align*}
  100PQF^2\mu^2 &\ge  \sum_{j=0}^{Q-1}\Delta(\phi_{i,j+1},\phi_{i,j})^2 
  \qquad && \text{(by \cref{eqn:sumPQ})} \\
  & \ge \frac{1}{Q} \Big (\sum_{j=0}^{Q-1}\Delta(\phi_{i,j+1},\phi_{i,j})\Big )^2
  \qquad && \text{(Cauchy--Schwarz)} 
  \\
  & \ge \frac{\Delta(\phi_{i,Q},\phi_{i,0})^2}{Q} \qquad && \text{(triangle inequality)} \\
  & \ge \frac{(R-L)^2}{Q}.
\end{align*}
Therefore, $(R-L)^2 \le 100 PQ^2F^2\mu^2 \le \rho^2$ as claimed.  
\end{proof}

The following lemma intuitively says that it is expensive to move all of the size-$c$ objects by a large total distance.
\begin{lemma}
  \label{lem:c-position-sum}
  Let $\lambda_i$ denote the sum of the locations of all the $Q$ size-$c$ objects in $\phi_{i,Q}$. Then,
for every $1\le i \le P$,
\[|\lambda_{i-1}-\lambda_{i}| \le \Delta(\phi_{i-1,Q},\phi_{i,Q})^2/\mu.\]

In particular, summing up this inequality and comparing with \cref{eqn:sumP} imply
  \[ |\lambda_0-\lambda_i| \le \sum_{i'=1}^{i} |\lambda_{i'-1}-\lambda_{i'}| \le 100 PF^2\mu \]
  for all $0\le i\le P$.
\end{lemma}
\begin{proof}
   Consider all the maximal changed intervals $[L_k,R_k)$ between $\phi_{i-1,Q}$ and $\phi_{i,Q}$. By \cref{lem:finger}, any  $[L_k,R_k)$ \emph{without} the finger objects must contain the same number of size-$c$ objects in $\phi_{i-1,Q}$ and $\phi_{i,Q}$. 
Since $\phi_{i-1,Q}$ and $\phi_{i,Q}$ have the same number of size-$c$ objects, by subtracting, we conclude that any $[L_k,R_k)$ (even if containing the finger objects) has the same number of size-$c$ objects in both memory states.  

   Let $q_1<q_2<\dots <q_Q$ and $q_1'<q_2'<\dots <q_Q'$ denote the locations of size-$c$ objects in $\phi_{i-1,Q}$ and $\phi_{i,Q}$, respectively.
   By the previous paragraph, the two objects at $q_j$ in $\phi_{i-1,Q}$ and at $q'_j$ in $\phi_{i,Q}$ either
   both belong to the same maximal changed interval, or are unchanged with $q_j=q'_j$.

We have 
   \[|\lambda_{i-1}-\lambda_i| = \left \lvert  \sum_{j=1}^Qq_j - \sum_{j=1}^{Q}q'_j\right \rvert \le \sum_{j=1}^Q| q_j-q'_j|.\]
For each non-zero term $|q_j-q'_j|$, if the two corresponding objects are in $[L_k,R_k)$, we bound this term by $|q_j-q'_j|\le R_k-L_k$.
Let $m_k$ denote the number of size-$c$ objects in $[L_k,R_k)$ in $\phi_{i-1,Q}$, which must satisfy $m_k \le (R_k-L_k)/c\le (R_k-L_k)/\mu$. Then,  we obtain the upper bound
 \[ \sum_{j=1}^Q|q_j-q'_j| \le \sum_{k} m_k (R_k - L_k) \le \sum_k (R_k-L_k)^2/\mu.\]
Since the total length of maximal changed intervals equals the difference between two memory states (\cref{defn:maximalchanged}), we have 
  \[  \sum_{k}(R_k-L_k)^2  \le \Big ( \sum_{k}(R_k-L_k)\Big )^2 = \Delta(\phi_{i-1,Q},\phi_{i,Q})^2.\]
  Chaining the three displayed inequalities finishes the proof.
\end{proof}

The previous two lemmas together imply that the finger object should always be confined to a small interval:
\begin{lemma}
   \label{lem:fingerclose}
  For every $0\le i\le P$,  $|p_{0,Q} - p_{i,Q}| \le 4\rho$.
\end{lemma}

\begin{proof}
By \cref{lem:c-clustered}, all size-$c$ objects in $\phi_{i,Q}$ are fully inside $[p_{i,Q}-\rho,\,p_{i,Q}+\rho)$. Similarly, all size-$c$ objects in $\phi_{0,Q}$ are fully inside $[p_{0,Q}-\rho,\,p_{0,Q}+\rho)$.

Suppose to the contrary that $|p_{0,Q}-p_{i,Q}|>4\rho$; here, assume $p_{0,Q}-p_{i,Q}>4\rho$ without loss of generality (the other case  follows similarly).
Then, for any two size-$c$ objects in $\phi_{0,Q}$ and in $\phi_{i,Q}$, at locations $q$ and $q'$ respectively, it holds that $q-q'\ge  (p_{0,Q}-\rho) - (p_{i,Q}+\rho) > 2\rho$.
Thus, \[\lambda_0-\lambda_i> Q\cdot 2\rho = 80\sqrt{P}Q^2F\mu.\]
On the other hand, \cref{lem:c-position-sum} states that \[|\lambda_0-\lambda_i| \le 100PF^2\mu.\]
Together, we get $20\sqrt{P}Q^2F\mu < 100PF^2\mu$, which contradicts our assumption \cref{eqn:assu1} that  
   $Q^2  \ge 5\sqrt{P}F$. Hence, we must have $|p_{0,Q}-p_{i,Q}|\le 4\rho$.
\end{proof}

Using \cref{lem:fingerclose}, the final plan is to show that it is expensive to insert all the $P$ size-$a$ objects, as a large fraction of them will end up far from the finger object.
To formalize this intuition, we need to define a suitable potential for the size-$a$ objects, and show that the total potential cannot change significantly.

\begin{definition}[$\Phi(\cdot)$]
   Define interval $[L^*,R^*)\coloneqq 
[p_{0,Q}-8\rho,\ p_{0,Q}+8\rho)$.
   If a size-$a$ object has location $q$, we say its \emph{potential} is 
   \[
     \Phi(q) \coloneqq \begin{cases}
       q-R^* & \text{if }q> R^*,\\
       L^*-q & \text{if }q\le L^*,\\
       0 & \text{otherwise.}
     \end{cases}
   \]
   Note that $\Phi(q)\ge 0$, and 
   \begin{equation}
      \label{eqn:phicontract}
|\Phi(q)-\Phi(q')|\le |q-q'|.
   \end{equation}
\end{definition}

\begin{lemma}
   \label{lem:a-sum}

  For $0\le i\le P$, let $\eta_i$ denote the sum of potentials $\Phi(\cdot)$ of all the size-$a$ objects in $\phi_{i,Q}$.
Then, $|\eta_{i-1} - \eta_{i}| \le
4\Delta(\phi_{i-1,Q},\phi_{i,Q})^2/\mu$.
\end{lemma}
\begin{proof}
    Consider all the maximal changed intervals $[L_k,R_k)$ between $\phi_{i-1,Q}$ and $\phi_{i,Q}$. By \cref{lem:finger}, any  $[L_k,R_k)$ \emph{without} the finger objects must contain the same number of size-$a$ objects in $\phi_{i-1,Q}$ and $\phi_{i,Q}$.  
    Let $[L_f,R_f)$ be the maximal changed interval that contains the finger objects in $\phi_{i-1,Q}$ and $\phi_{i,Q}$, whose locations are $p_{i-1,Q}$ and $p_{i,Q}$ respectively. 
    Since $\phi_{i,Q}$ has one more size-$a$ object than $\phi_{i-1,Q}$, we know $[L_f,R_f)$ contains one more size-$a$ object in $\phi_{i,Q}$ than in $\phi_{i-1,Q}$.
    
    Let $q_1<q_2<\dots < q_r$ and $q'_1<q'_2<\dots < q'_r$ denote the locations of size-$a$ objects that are outside $[L_f,R_f)$ in $\phi_{i-1,Q}$ and $\phi_{i,Q}$, respectively. Let $q_{r+1}<q_{r+2}<\dots< q_{i-1}$ and $q'_{r+1}<q'_{r+2}<\dots< q'_i$ denote the locations of size-$a$ objects inside $[L_f,R_f)$ in $\phi_{i-1,Q}$ and $\phi_{i,Q}$, respectively.
   By the previous paragraph, for $j\in [r]$, the two objects at $q_j$ and at $q'_j$  either
   both belong to the same maximal changed interval $[L_k,R_k)$ (for some $k\neq f$), or are unchanged with $q_j=q'_j$.

   To bound $|\eta_{i-1}-\eta_i|$, we write
   \begin{align}
   |\eta_{i-1}-\eta_i| = \left \lvert  \sum_{j=1}^{i-1} \Phi(q_j) - \sum_{j=1}^{i}\Phi(q'_j)\right \rvert &\le \sum_{j=1}^r| \Phi(q_j)-\Phi(q'_j)| + \sum_{j=r+1}^{i-1}\Phi(q_j) + \sum_{j=r+1}^{i}\Phi(q'_j),\nonumber \\
   &\le \sum_{j=1}^r| q_j-q'_j| + \sum_{j=r+1}^{i-1}\Phi(q_j) + \sum_{j=r+1}^{i}\Phi(q'_j),
   \label{eqn:etabound}
   \end{align}
   where we used the non-negativity of $\Phi(\cdot)$ and \cref{eqn:phicontract}.
   
   To bound the first sum in \cref{eqn:etabound}, we proceed in the same  way as in the proof of \cref{lem:c-position-sum}.
For each non-zero term $|q_j-q'_j|$, if the two corresponding objects are in $[L_k,R_k)$, we bound this term by $|q_j-q'_j|\le R_k-L_k$.
Let $m_k$ denote the number of size-$a$ objects in $[L_k,R_k)$ in $\phi_{i-1,Q}$, which must satisfy $m_k \le (R_k-L_k)/a \le  (R_k-L_k)/\mu$. Then, we obtain the upper bound
 \[ \sum_{j=1}^r|q_j-q'_j| \le \sum_{k\neq f} m_k (R_k - L_k) \le \sum_{k\neq f} (R_k-L_k)^2/\mu.\]

   Now, we bound the last two sums in \cref{eqn:etabound}.
   We only need to consider the case where they contain at least one positive term, i.e., there exists some $q^*\in \{q_{r+1},\dots,q_{i-1}\}\cup \{q_{r+1}',\dots,q_i'\}$ such that $\Phi(q^*)>0$. By definition of $\Phi(\cdot)$, this means $q^* \notin [L^*,R^*) = [p_{0,Q}-8\rho,\ p_{0,Q}+8\rho)$, so $|q^*-p_{0,Q}|\ge 8\rho$.  Since both $q^*$ and $p_{i,Q}$ are in the interval $[L_f,R_f)$, we have \begin{equation} \label{eqn:rllarge} R_f-L_f \ge  |q^*-p_{i,Q}| \ge |q^*-p_{0,Q}|-|p_{i,Q}-p_{0,Q}| \ge 8\rho - 4\rho = 4\rho, \end{equation} where in the last inequality we used \cref{lem:fingerclose}.
On the other hand, for any $q\in \{q_{r+1},\dots,q_{i-1}\}\cup \{q_{r+1}',\dots,q_i'\} \subset [L_f,R_f)$,  we have
 \begin{align*}
 \Phi(q) & \le |q-p_{0,Q}| + \Phi(p_{0,Q})    \qquad && \text{(by \cref{eqn:phicontract})} \\
  &  =  |q-p_{0,Q}|  \qquad && \text{(by definition of $\Phi(\cdot)$)}\\
  &  \le |q-p_{i,Q}|+|p_{i,Q}-p_{0,Q}|\\
  & \le (R_f-L_f)+ 4\rho\qquad && \text{(by $q,p_{i,Q}\in [L_f,R_f)$, and \cref{lem:fingerclose})}\\
  & \le 2(R_f-L_f).\qquad && \text{(by  \cref{eqn:rllarge})}
 \end{align*} 
Since there are at most $(R_f-L_f)/a $ size-$a$ objects in $[L_f,R_f)$ in $\phi_{i-1,Q}$ (or, in $\phi_{i,Q}$),
we conclude that the second sum and the third sum of \cref{eqn:etabound} each can be upper-bounded by $2(R_f-L_f)\cdot (R_f-L_f)/a\le 2(R_f-L_f)^2/\mu$.

Summing up, \cref{eqn:etabound}  can be bounded as
\begin{align*}
|\eta_{i-1}-\eta_i|&\le  \sum_{k\neq f}(R_k-L_k)^2/\mu + 2(R_f-L_f)^2/\mu +2(R_f-L_f)^2/\mu \\
& \le 4\sum_{k}(R_k-L_k)^2/\mu \\
& \le 4\Big (\sum_{k}(R_k-L_k)\Big)^2/\mu \\
& = 4\Delta(\phi_{i-1,Q},\phi_{i,Q})^2/\mu
\end{align*}
 as claimed.
\end{proof}

Now we are ready to finish the proof.

\begin{proof}[Proof of \cref{lem:goal}]
   Recall $\eta_i$ is the sum of potentials $\Phi(\cdot)$ of all the size-$a$ objects in $\phi_{i,Q}$. We have
  \begin{align*}
   \eta_P &\le \eta_0 +  \sum_{i=1}^P|\eta_{i-1} - \eta_{i}| \\
    & \le \eta_0 +  \sum_{i=1}^P4\Delta(\phi_{i-1,Q},\phi_{i,Q})^2/\mu \qquad && \text{(by  \cref{lem:a-sum})}\\
    & \le \eta_0 + 400PF^2\mu\qquad && \text{(by  \cref{eqn:sumP})}\\
    & = 400PF^2\mu. \qquad && \text{($\phi_{0,Q}$ contains no size-$a$ objects)}
  \end{align*} 

By definition of $\Phi(\cdot)$, if a size-$a$ object is not fully contained in the interval $[L^*-10\rho,R^*+10\rho + a) = [p_{0,Q}-18\rho,\ p_{0,Q} + 18\rho + a)$, then its potential is at least $10\rho$.  The number of size-$a$ objects in $\phi_{P,Q}$ that are fully contained in this interval is at most 
\[(36\rho+a)/a \le 36\rho/\mu + 1 \underbrace{<}_{\text{by \cref{defn:rho}}} 400\sqrt{P}QF \underbrace{\le}_{\text{by \cref{eqn:assu3}}} P/2. \]
   Therefore, among the $P$ size-$a$ objects in $\phi_{P,Q}$, at least $P-P/2=P/2$ of them  have potential at least $10\rho$, giving $\eta_{P}\ge (P/2)\cdot (10\rho) = 50P^{1.5}QF\mu$.  
   By comparing with the upper bound $\eta_P \le 400PF^2\mu$ from the previous paragraph, we get $\sqrt{P}Q \le 8F$. In particular, $Q\le 8F$, but this contradicts $Q = 4000 F^2$ by \cref{eqn:assu4} and $F\ge 1$. This finishes the proof of \cref{lem:goal} (which implies \cref{thm:squaredlowerbound-discrete} and \cref{thm:mainsecondmoment}). 
\end{proof}

\section{Conclusion and open questions}
\label{sec:con}

Our upper bound result demonstrates the surprising power of the sunflower lemma for the Memory Reallocation problem.
Our lower bound results rely on the additive structure of the object sizes in the designed hard instances.
In general, we believe that  applying additive combinatorial techniques to various resource allocation and scheduling problems is a  fruitful direction for future research.
\newline

We conclude with several concrete open questions:
\paragraph*{Closing the gap.}
    One of the main open questions is to narrow the gap between our $O(\log^4 \epsilon^{-1}\cdot (\log \log \epsilon^{-1})^2)$ upper bound and the $\Omega(\log\eps^{-1})$ lower bound for the expected reallocation overhead.  We conjecture that the lower bound is closer to the truth.
\paragraph*{Time complexity.}
 Kuszmaul \cite{Kuszmaul23} provided a time-efficient implementation of his allocator for tiny objects.
 It would be interesting to make our allocator time-efficient as well.

 A main obstacle towards time-efficiency is the lack of a fast algorithmic version of \cref{lem:grouping-temp}, which was proved via 
 Erd\H{o}s and S\'{a}rk\"{o}zy's non-constructive argument based on the sunflower lemma. 
 As a preliminary attempt, one could compute the construction of \cref{lem:grouping-temp} (with slightly worse logarithmic factors) in $\poly(n,w)$ time, by combining Erd\H{o}s--Rado's proof of the sunflower lemma with the standard dynamic programming algorithm for the counting version of the Subset Sum problem. However, this is too slow for our application, where $w$ is as large as $\eps^{-1}$.
  
 After completing this paper, we became aware of a recent paper by  Chen, Mao, and Zhang \cite{soda26chen},  which, despite the different motivation, arrived at the same question and made significant progress.
More specifically, \cite[Theorem 1.6]{soda26chen} implies that, for any $\delta\in (0,1)$, the construction of our \cref{lem:grouping-temp} can be implemented in $\widetilde O(n + w^\delta)$ time, at the cost of worsening all the $\log w$ factors in the lemma statement to $(\log w)^{O(1/\delta)}$. 
This suggests the possibility of designing an allocator with a time complexity overhead factor of $\exp(O(\sqrt{\log \eps^{-1}\log\log \eps^{-1}}))$. 

  \paragraph*{High-probability guarantee.}
Although  polylogarithmic overhead with high probability is impossible in general (by \cref{thm:mainsecondmoment}), it is still achievable in some interesting special cases.
For example, one can show that Kuszmaul's allocator \cite{Kuszmaul23} for power-of-two sizes actually achieves $O(\log(1/\eps))$ overhead with high probability in $1/\eps$.
We conjecture that high-probability polylogarithmic overhead is also achievable for tiny object sizes (for example, less than $\eps^4 M$).

Another interesting direction is to design allocators achieving $O(1/\eps^\alpha)$ overhead with high probability in $1/\eps$ (or even deterministically) with small exponent $\alpha$. Our lower bound (\cref{thm:mainsecondmoment}) implies $\alpha \ge 1/14$, but the best known upper bound is only $\alpha = 1$, achieved by the folklore deterministic allocator.
\paragraph*{Other settings.}
It would also be interesting to further investigate Memory Reallocation in the cost-oblivious setting of \cite{BenderFFFG17},  or the request fragmentation setting of \cite{nicole}.

\subsection*{Acknowledgements}
I would like to thank Nathan Sheffield and Alek Westover for inspiring discussions that sparked my interest in this problem.
Additionally, I thank Shyan Akmal, Lin Chen, Hongbo Kang, Tsvi Kopelowitz, Mingmou Liu, Jelani Nelson, Kewen Wu, and Renfei Zhou for helpful discussions, and Nicole Wein for sharing a manuscript of~\cite{nicole}.

	\bibliographystyle{alphaurl} 
	\bibliography{main}
   
\end{document}